\documentclass[journal,10pt]{IEEEtran}
\usepackage{mathrsfs}
\usepackage{amsmath}
\usepackage{amssymb}
\usepackage{amsthm}
\usepackage[noadjust]{cite}
\usepackage{cite}
\usepackage{eurosym}
\usepackage{graphicx}
\usepackage{epstopdf}
\usepackage{tikz,lipsum}
\usepackage{subcaption}
\usepackage{float}
\usepackage[T1]{fontenc}
\usepackage{ragged2e}
\usepackage{multicol}
\usepackage{listings}
\usepackage{algorithm,algorithmic}
\usepackage{authblk}
\usepackage{times}
\usepackage{array}
\usepackage[left=1.5cm,top=1.5cm,right=1.5cm,bottom=1.5cm]{geometry}

\newtheorem{theorem}{Theorem}

\newtheorem{lemma}{Lemma}

\newtheorem{remark}{Remark}
\newcounter{mytempeqncnt}

\begin{document}
\title{A PHY Layer Security of a Jamming-Based Underlay Cognitive Hybrid Satellite-Terrestrial Network}
\author{Mounia~Bouabdellah, and Faissal~El~Bouanani,~\IEEEmembership{Senior~Member,~IEEE}
\thanks{M. Bouabdellah and F. El Bouanani are with ENSIAS College of Engineering, Mohammed V University, Rabat, Morocco (e-mails: \{mounia\_bouabdellah, {faissal.elbouanani}\}@um5.ac.ma).}% <-this % stops a space
\thanks{Manuscript received XXX, XX, 2015; revised XXX, XX, 2015.}}

\markboth{IEEE Transactions on Vehicular Technology,~Vol.~XX, No.~XX, XXX~2015}
{}
%{Shell \MakeLowercase{\textit{et al.}}: Bare Demo of IEEEtran.cls for Journals}
\maketitle
\begin{abstract}
In this work, we investigate the physical layer security of a jamming-based
underlay cognitive hybrid satellite-terrestrial network {consisting of a}
radio frequency link at the first hop and an optical feeder at the second
hop. Particularly, one secondary user (SU) is transmitting data {to an end-user} optical ground station {($\boldsymbol{D}$)} through the aid of a relay satellite, in the presence of {an active} eavesdropper 
{at} each hop. {Moreover}, another SU located in the first hop is acting as a friendly jammer and continuously broadcasting an artificial noise that cannot be decoded by the wiretapper so as to impinge positively on the system's secrecy. Owing to the underlying strategy, the SUs are permanently adjusting their transmit powers in order to avoid causing harmful interference to primary users. The RF channels undergo shadowed-Rician and Rayleigh fading models, while the optical link is subject to Gamma-Gamma {turbulence with pointing error}. Closed-form and asymptotic expressions for the intercept probability (IP) are derived considering two different scenarios regardless of the channel's conditions, namely (i) absence and (ii) presence of a friendly jammer. The effect of various key parameters on IP, e.g., {sources'} transmit power, {artificial noise}, maximum tolerated interference power, and fading severity parameters are examined. Precisely, we aim to answer the following question: could a friendly jammer further enhance the security of such a system even in a low SNR regime? All the derived results are corroborated by Monte Carlo simulations and new insights into the considered system's secrecy are gained.
\end{abstract}
\IEEEpeerreviewmaketitle
\section{Introduction}

With the arrival of the fifth-generation (5G) wireless cellular network and
the increasing number of smart mobile devices along with the corresponding
Internet of Things (IoT) applications, the demand for high-speed
communications is increasing as well. {Moreover, providing connectivity to rural areas and achieving worldwide connectivity is expected to be accomplished in sixth-generation (6G) communication networks \cite{Dang}, \cite{Yaacoub}}. In this context, hybrid satellite-terrestrial systems have been widely advocated among appropriate schemes to enlarge the network's coverage {in 6G networks}. Such networks operate in 
{a} conventional Ku (12 GHz) and Ka (26.5-40 GHz) bands.
Their aim {consists of} ensur{ing} a higher data rate (e.g., tens to
hundreds of Megabits per second) per end-user, and provide high-speed
communications to users in particular areas characterized by a low quality of
service when traditional terrestrial networks are employed. With the
objective of fulfilling the aforementioned data rate requirements, hundreds
of Gbps of aggregate throughput link should be allocated for the
ground-satellite connection. Indeed, one of the omnipresent issues in
setting up such high-speed networks {lies in} bandwidth limitation {of a}
radio-frequency (RF) spectrum, particularly in the Ku and Ka{-}bands, which
restricts the coverage of such systems to a limited number of end-users \cite%
{Gharanjik2013}. In this regard, hybrid
satellite-terrestrial cognitive network (HSTCN), employed jointly with free-space optical (FSO) technology could provide
an efficient solution to alleviate the spectrum scarcity issue.

Remarkably, HSTCN has received considerable attention during the last
years. For instance, the project cognitive radio for satellite
communications (CoRaSat) \cite{corasat2012} has proposed the implementation
of dynamic and smart spectrum usage techniques for satellite communications
(SatCom) in {view} to exploit either the underused or unused frequency
resources that have been already assigned to other services. Furthermore,
high-speed communication {(e.g., Tbps per optical beam)} can be achieved by
employing FSO technology {\cite{Esmai2015}} {where} the data
is transmitted with the help of {an} optical source emitting light beams in
either visible (400-800 nm) or infrared (1500-1550 nm) spectrum bands.

Interestingly, FSO communication is {most immune} to the
interference and provide a high level of security
against wiretapping attacks due to its narrow beamwidth.
In contrast, the broadcast nature of RF wireless
communication makes it vulnerable to eavesdropping attacks. Moreover,
allowing the unlicensed users to share the spectrum with the licensed ones
in CRNs can increase the risk of {overhearing} the %
{signals of legitimate transmitters}. In this context,
physical layer security (PLS) has been proposed as an efficient solution to
support and supplement existing cryptography protocols that are used in
higher {network layers}. Its concept was first introduced in
the seminal work of Wyner \cite{Wyner}{,} therein, it has been shown that
better secrecy can be achieved if the legitimate {channel's}
capacity exceeds 
{that of the eavesdropper one by a certain
	rate threshold}. Towards this end, the main concern of PLS is to strengthen
the secrecy capacity and therefore {to improve the SS}.
Actually, various techniques have been suggested to enhance the {SS},
namely, {artificial noises}, relay selection protocols,
energy harvesting, cooperative communications, spatial diversity %
{and} combining techniques, 
{massive
	multiple-input and multiple-output}, non-orthogonal multiple access,
zero-forcing precoding techniques, etc.

%%%%%%%%%%%%%%%%%%%%%%%%%%%%%%%%%%%%%%%%%%%%PLS of CRN 
{Recently, several works have been conducted into the secrecy
	analysis of CRNs with various system strategies} \cite{lei2016secrecy}-\cite%
{AccessMounia}. For instance, non-cooperative CRNs 
{in
	the presence of a direct link} between the source and the destination %
{have} been considered in \cite{lei2016secrecy}-\cite%
{nguyen2017secure}. Specifically, it has been assumed in \cite%
{lei2016secrecy} that {both} sources and receivers are
equipped with multiple antennas, whereas in \cite{tran2017cognitive}-\cite%
{nguyen2017secure} multiple antennas are considered %
{exclusively} at the receivers. The PLS of dual-hop
relay-assisted CRNs was investigated in \cite{sakran2012proposed}-\cite%
{lei2017secrecy}. Explicitly, multiple and single eavesdropper have been
considered {for various relay selection protocols} in \cite%
{sakran2012proposed},\cite{ding2016secrecy} and \cite{ho2017analysis},\cite%
{lei2017secrecy}, respectively. {Moreover}, {the secrecy
	outage probability (SOP) was derived by {examining two}
	scenarios, namely} the {selection of the relay} maximizing
the achievable secrecy rate was {presented} in \cite%
{sakran2012proposed}, while {the wiretap link signal-to-noise (SNR)
	minimization-based relay node selection protocol was} %
{analyzed} in \cite{ding2016secrecy},\cite{ho2017analysis}.
Significantly, the impact of jamming signals on the {SS} under various
jammers' selection protocols has been discussed in \cite{zou2016physical}-%
\cite{MouniaBouabdellah}. {By} leveraging the powerfulness of
EH in enhancing the reliability of wireless communication systems, the PLS
of {EH-based CRNs} was inspected by assuming %
{either} time-switching {or} power-splitting
protocols in \cite{TVT2} {and} \cite{tgcn}-\cite{AccessMounia},
respectively. {Furthermore}, the energy was considered to be
harvested from the PU's and SU's signals in \cite{TVT2}-\cite{tgcn} and \cite%
{ISWCS}-\cite{AccessMounia}, respectively.

\subsection{Motivation and contributions}

To the best of the authors' knowledge, very few research works have
investigated the PLS of HSTCNs. 
{The secrecy analysis of
	HSTCN where the terrestrial secondary network is sharing the spectrum with a
	satellite system, acting as a primary network is examined in
	\cite{An2016JSAC} by maximizing the SUs rate and considering various
	beamforming techniques.} 
{Distinctively, the authors in
	\cite{LU2018access} dealt with the minimization of the total transmit power
	of numerous terrestrial base stations and on board the satellite subject to
	the PU secrecy rate constraint.} Later, the average secrecy capacity (ASC)
and SOP of a downlink hybrid satellite-FSO cooperative system were derived
in \cite{Ai2019} by considering both amplify-and-forward and
decode-and-forward (DF) relaying protocols. Likewise, the authors of \cite%
{Illi2020} investigated the PLS of a hybrid very high throughput satellite
communication system with an FSO feeder link. 
{Therein, a
	satellite combines the received data from multiple optical ground stations,
	performs a decoding process, regenerates the information signal, and
	forwards it to the end-user with zero-forcing precoding so that to cancel
	the interbeam interference at the receivers.}

Motivated by the above, we {aim at this work to investigate}
the PLS of HSTCNs consisting of one secondary user (SU) source communicating
with a secondary base station through the aid of relay satellites in
presence of two eavesdroppers that are overhearing both communication hops.
Besides, another SU acts as a friendly jammer at the first %
{hop and} broadcasts an artificial noise {to}
disrupt the eavesdropper. Without loss of generality, 
{the RF
	secondary and primary links are undergoing} shadowed-Rician and Rayleigh
fading models, respectively, while the FSO links employed in the second hop
are {subject to} Gamma-Gamma (GG) fading with pointing error
for a DF relaying system, and in the presence of two eavesdroppers at both
communication hops, deriving the IP is not straight forward. Therefore, the
main aim of this work is to derive a new formula for the IP that considers
the presence of two eavesdroppers for a DF relaying system.

Pointedly, the main contributions of this work can be summarized as follows:

\begin{itemize}
	\item A novel expression for the IP of a dual-hop 
	DF relaying in the presence of eavesdroppers {at each hop} is
	derived.
	
	\item Capitalizing on the above result, the IP expression of
		HSTCN is derived in closed-form for both the presence and absence of friendly
		jammer scenarios.
	
	\item 
	{The asymptotic expression for the IP in high SNR regime
		is also provided, based on it, the achievable diversity order is retrieved.}
	
	\item Insightful discussions on the impact of different key parameters of %
	{HSTCN on its security} are also provided. Specifically, we
	demonstrated that the {SS} can be enhanced by increasing (i) SU's transmit
	power, (ii) maximal tolerated interference power {(MTIP)} at
	PU, {and} (iii) average power of the downlink channel along
	with the satellite's transmit power. Moreover, we %
	{demonstrated} that under low {source}
	transmit power and low {MTIP constraints}, the friendly
	jammer does not enhance the SS.
\end{itemize}

\subsection{Organization of the paper}

The remainder of this paper is organized as follows. The considered HSTCNs
is presented in Section II, while in Section III, 
{a novel
	framework for the IP computation of a dual-hop DF-based CRN is provided, and
	the closed-form expressions for the IP of HSTCNs along with the asymptotic
	analysis are presented}. The numerical and simulation results are presented
and discussed in Section IV. Finally, the last section reports closing
remarks that summarize the current contribution.

\subsection{Notations}

For the sake of clarity, the different notations used throughout the paper
are defined and summarized in Table \ref{nota}.
\begin{figure*}[tbp]
	\begin{table}[H]
		\caption{List of functions and symbols.} 
		\begin{centering}
			\begin{tabular}{c|c|c|c} 
				\hline
				\hline
				\textit{Symbol} & \textit{Meaning} & \textit{Symbol} & \textit{Meaning} \tabularnewline
				\hline
				{$h_{.}$} & {Fading amplitude} & {$\mathbb{E}\left[ .\right] $} & {Expectation operator} \tabularnewline
				\hline
				$g_{.}=\left|h_{.}\right|^2 $ & Channel gain & {$ F_{.}\left( .\right) $} &{Cumulative distribution function (CDF)}\tabularnewline
				\hline
				$P^{\max}_{\mathbf{Q}}$ &  Maximum transmit power at $ \mathbf{Q}\in\{\mathbf{S},\mathbf{S_{J}}\} $ & {$F_{.}^{c}\left(.\right)$} &  {Complementary CDF} \tabularnewline
				\hline
				$ P_{I} $ & Maximum tolerated interference power at PU$_{Rx} $& {$f_{.}\left( .\right) $} &  {Probability density function (PDF)}\tabularnewline
				\hline
				$P_{Tx}$ & Transmit power of node $Tx\in\{\mathbf{S},\mathbf{S_{J}},\mathbf{R}\}$ & {$\gamma _{inc}\left(.,.\right) $} & Lower incomplete Gamma function \cite[Eq. (8.350.1)]{Table} \tabularnewline
				\hline
				$x_{Tx}$ & Transmitted signal from the node $Tx$ & {$ \Gamma\left(.,. \right)  $} & {Upper incomplete Gamma function \cite[Eq. (8.350.2)]{Table}} \tabularnewline
				\hline
				$h_{X}$  & Fading amplitudes, $X\in \{SR,$ $SE_{1,}$ $S_{J}E_{1}\}$ & $ I_{Z} $ & Channel irradiance\tabularnewline
				\hline
				$n_{Rx}$ & Additive white Gaussian noise(AWGN) at $Rx\in\{\mathbf{R},$ $\mathbf{D},$ $\mathbf{E}_{1},$ $\mathbf{E}_{2}\}$ & {$\gamma_{th} $} & {Decoding threshold SNR}  \tabularnewline 
				\hline
				$\omega _{Z}$&  Portion of power received by $\mathbf{Z}${'s} photodetector, $\mathbf{Z}\in\{\mathbf{D},\mathbf{E_{2}}\}$
				& {$m_{X}$}& {Fading severity parameter} \tabularnewline
				\hline
				$\eta$ & Optical-to-electrical conversion ratio
				& {$ b_{X} $} &  {Half average power of the multi-path component}\tabularnewline
				\hline
				$ N $ & AWGN power & {$\Omega_{X}$} &  {Average power of LOS component} \tabularnewline
				\hline
				$d_{Z}$ & Distance between the satellite and the node $\mathbf{Z}$& {$ r $} &  {Detection technique parameter}\tabularnewline
				\hline
				$\phi $ & Path{-}loss exponent& $ \alpha_{Z},\beta_{Z} $ & Turbulence-induced fading parameters \tabularnewline
				\hline
				$\gamma_{R}$ & SNR at $\mathbf{R}$ & $\gamma_{E_{1}}$ & SNR at $\mathbf{E_{1}}$ in the absence of a jammer \tabularnewline
				\hline
				$\gamma^{(J)}_{E_{1}}$ & SNR at $\mathbf{E_{1}}$ in the presence of a jammer &$\gamma_{Z}$& SNR at node $\mathbf{Z}$ \tabularnewline
				\hline
				$ G^{m,n}_{p,q} $& Meijer's G-function \cite[Eq. (9.301)]{Table} & $\mathcal{U}_{Q}$ &  SNR at $\mathbf{E_1}$ \tabularnewline
				\hline
				\hline
			\end{tabular}
			\par
		\end{centering}
		\label{nota}
	\end{table}
\end{figure*} 
\section{System and Channel Models}
We consider an underlay cognitive satellite system, presented in Fig. 1, consisting of two SU sources, namely a data source $\mathbf{S}$, a jammer source $\mathbf{S}_{J},$ satellite $\mathbf{R}$ that serves as a relay, one optical ground station $\mathbf{D}$, two eavesdroppers $\mathbf{E}_{1}$ and $\mathbf{E}_{2}$ intercepting the communication at the first and the second hop, respectively, one primary transmitter PU$_{T_{x}}$, and one primary receiver PU$_{R_{x}}$. It is worth mentioning that the jamming signal sent by $\mathbf{S}_{J}$ can be canceled at the legitimate receiver (i.e. $\mathbf{R}$), while the eavesdropper $\mathbf{E}_{1}$ is not able to decode it. Without loss of generality, the source communicates with the optimal relay through an RF-link, while at the second hop $\mathbf{R}$ forwards data to $\mathbf{D}$ through an optical feeder link. 
\begin{figure}[tbp]
	\centering
	\includegraphics[scale=0.38]{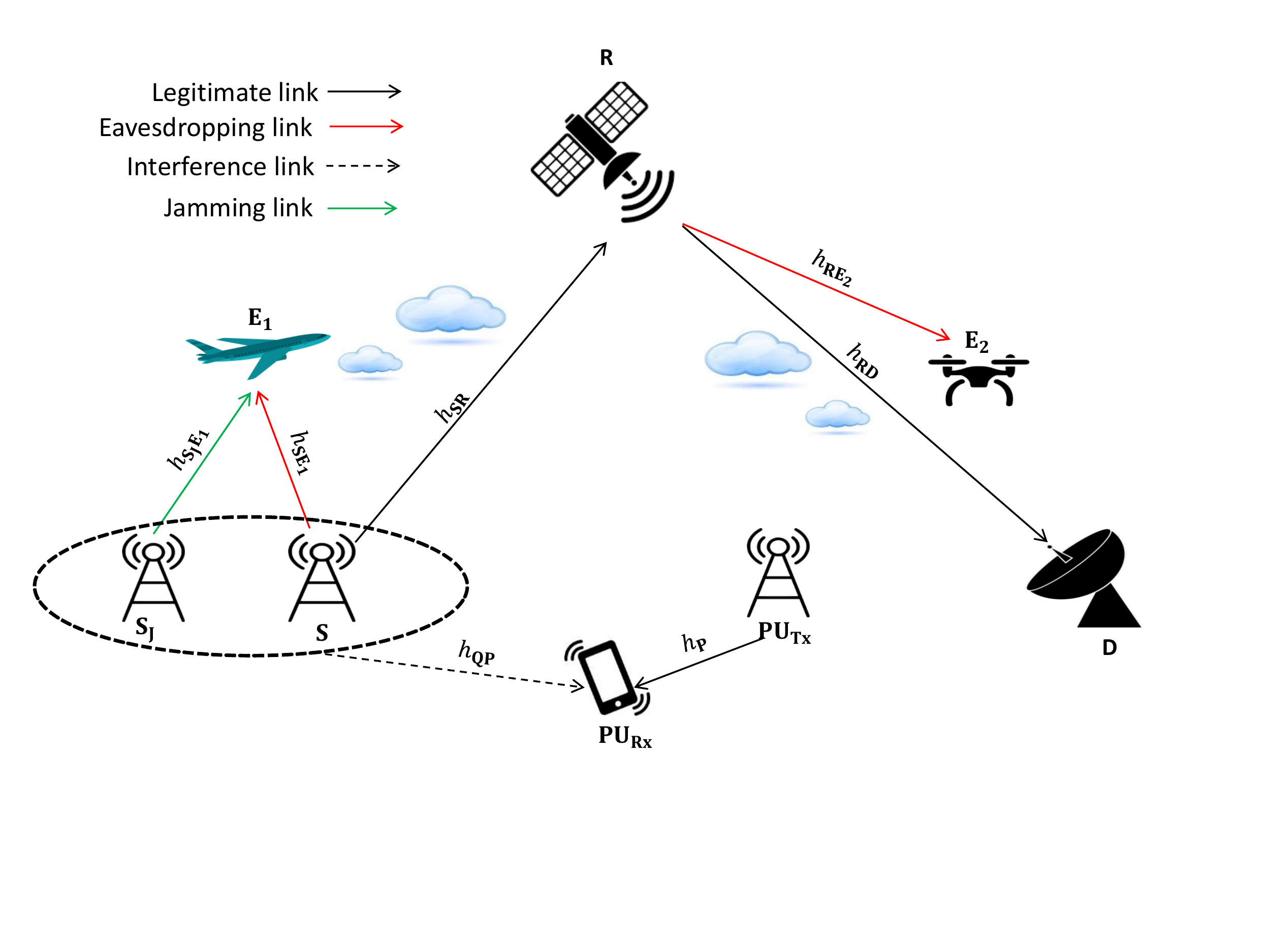}
	\caption{System model.}
\end{figure}
The received signal at $\mathbf{R}$ and $\mathbf{E}_{1}$ are, respectively, given by 
\begin{equation}
y_{R}=\sqrt{P_{S}}x_{S}h_{SR}+n_{R},  \label{Y_R}
\end{equation}
\begin{equation}
y_{E_{1}}=\sqrt{P_{S}}x_{S}h_{SE_{1}}+\varepsilon \sqrt{P_{S_{J}}}%
x_{S_{J}}h_{S_{J}E_{1}}+n_{E_{1}},
\end{equation}%
while, the received signal at node {$\mathbf{Z}\in \left\{ \mathbf{D},%
	\mathbf{E}_{2}\right\} $} is given by 
\begin{equation}
y_{Z}=\sqrt{P_{R}\omega _{Z}}\left( \eta I_{Z}\right) ^{\frac{r}{2}%
}x_{R}+n_{Z},  \label{Y_D}
\end{equation}%
where%
{, $\varepsilon$ is either equal to $ 0 $ or to $ 1 $ in
	absence or presence of artificial noise, respectively.} $P_{Tx}$, $x_{Tx}$%
, $n_{Rx},$ $\eta ,$ $\omega _{Z}$ are defined in Table \ref{nota}, $I_{Z}$
is the irradiance of the link $\mathbf{R}$-$\mathbf{D}$, $\left( {i.e.,}%
I_{Z}=I_{Z}^{(a)}I_{Z}^{(p)}I_{Z}^{(l)}\right) $ defined as the product of
the irradiance fluctuation caused by atmospheric turbulence, the pointing
error due to the beam misalignment, and the free-space path loss (FSPL),
respectively. {Of note, the latter irradiance can be expressed as} $%
I_{Z}^{(l)}=I_{t}e^{-\phi d_{Z}},$ with $I_{t}$ stands for the laser
emittance. 
{For simplicity reasons, the FSPL’s irradiance
	is considered normalized to unity ($ \phi=0 $)}. Furthermore, $r$ refers to
the detection technique index (i.e., $r=1$ for coherent detection and $r=2$
for direct detection).

To avoid interference with the PU signal, the SNR at $\mathbf{R}$ can be
characterized as 
\begin{equation}
\gamma _{R}=\min \left( \overline{\gamma }_{S},\frac{\overline{\gamma }_{I}}{%
	g_{SP}}\right) g_{SR},  \label{GammaR}
\end{equation}%
while the SNR at $\mathbf{E}_{1}$ {in the case of} absence and presence of a
friendly jammer are, respectively, given by 
\begin{equation}
\gamma _{E_{1}}={\mathcal{U}_{S}},  \label{gammaE1_without_jammer}
\end{equation}
and 
\begin{equation}
\gamma _{E_{1}}^{\left( J\right) }={\frac{\mathcal{U}_{S}}{\mathcal{U}%
		_{S_{J}}+1}},  \label{gammaE1}
\end{equation}
where 
\begin{equation}
{\mathcal{U}_{Q}=\min \left( \overline{\gamma }_{Q},\frac{\overline{\gamma }%
		_{I}}{g_{QP}}\right) g_{QE_{1}},\ \mathbf{Q}\ \in \left\{ \mathbf{S},\mathbf{%
		S}_{J}\right\} ,}  \label{UQ}
\end{equation}%
$\overline{\gamma }_{Q}=P_{Q}^{\max }/N,$ $\ $ and $\overline{\gamma }
_{I}=P_{I}/N$, 
{with $ N $ presents the thermal power noise
	at the receivers, assumed identical.}

Furthermore, the SNR at node $\mathbf{Z}\in\left\{\mathbf{D},\mathbf{E}%
_{2}\right\} $ can be straightforwardly expressed from (\ref{Y_D}) as 
\begin{equation}
\gamma_{Z}=\frac{P_{R}\omega_{Z}\left( \eta I_{Z}\right)^{r}}{N}.
\label{gammaZ}
\end{equation}

Given that the satellite {$\mathbf{R}$} performs {the} DF protocol, the
equivalent SNR of the end-to-end link is given by 
\begin{equation}
\gamma_{eq}=\min \left( \gamma_{R},\gamma _{D}\right).
\label{gamma_equivalent}
\end{equation}
All fading amplitudes are assumed to be independent and identically distributed (i.i.d). Specifically, the amplitudes of the terrestrial (i.e. $\mathbf{S}$-PU$_{Rx}$ and $\mathbf{S}_{J}$-PU$_{Rx}$) and uplink channels (i.e. $\mathbf{S}$-$\mathbf{R}$, $\mathbf{S}$-$\mathbf{E}_{1},$ and $\mathbf{S}_{J}$-$\mathbf{E}_{1}$) are Rayleigh and shadowed-Rician distributed, respectively. Therefore, the PDF of the channel gain corresponding to the latter channel can be characterized as 
\begin{align}
f_{g_{X}}\left( x\right) & =\Delta _{X}e^{-\beta _{X}x}\text{ }
_{1}F_{1}\left( m_{X};1;\delta _{X}x\right)  \label{pdf_hop1} \\
& \overset{(a)}{=}\Delta _{X}e^{-\upsilon _{X}x}\sum_{n=0}^{m_{X}-1}\phi
_{X}^{(n)}x^{n},\text{ }X\in \{SR,SE_{1,}S_{J}E_{1}\},  \notag
\end{align}
where $\Delta _{X}=\frac{1}{2b_{X}}\left( \frac{2b_{X}m_{X}}{2b_{X}m_{X}+\Omega_{X}}\right) ^{m_{X}},$ $\upsilon_{X}=\beta_{X}-\delta_{X}$ $,$ $\beta _{X}=\frac{1}{2b_{X}},$ $\delta_{X}=\frac{\beta_{X}\Omega_{X}}{2b_{X}m_{X}+\Omega _{X}},$ $\phi _{X}^{(n)}=\frac{\left(m_{X}-1\right) !\delta _{X}^{n}}{\left( m_{X}-1-n\right) !\left(n!\right)^{2}}$, 
$_{1}F_{1}(.;.;.)$ denotes the confluent hypergeometric function \cite[Eq. (9.210)]{Table}, and step ($a$) follows by assuming that $m_{X}$ is a positive-valued number and by using jointly Eqs. (06.10.02.0003.01) and (07.20.03.0025.01) of \cite{Wolfram}. The corresponding CDF can be straightforwardly obtained from the above PDF as 
\begin{eqnarray}
F_{g_{X}}\left( x\right) &=&\Delta _{X}{}\sum_{n=0}^{m_{X}-1}\phi_{X}^{(n)}\int_{0}^{x}t^{^{n}}e^{-\upsilon _{X}t}dt  \label{cdf_hop1} \\
&=&\Delta _{X}{}\sum_{n=0}^{m_{X}-1}{\frac{\phi _{X}^{(n)}}{\upsilon_{X}^{n+1}}\gamma _{inc}\left( {\small n+1},\upsilon _{X}x\right)}
{,}  \notag
\end{eqnarray}
where $ \gamma_{inc}\left(.,.\right) $ denotes the lower incomplete Gamma function \cite[Eq. (8.350.1)]{Table}.

On the other hand, as the atmospheric turbulence induced-fading with pointing error for the received optical beam is modeled with Gamma-Gamma distribution, the PDF and CDF of the SNR $\gamma_{Z}$ are expressed, respectively as \cite{zedini} 
\begin{align}
f_{\gamma _{Z}}(z)& =\frac{\mathcal{O}_{Z}}{rz}G_{1,3}^{3,0}\left(\Upsilon
_{Z}\sqrt[r]{\frac{z}{\mu _{r}^{(Z)}}}\left\vert 
\begin{array}{c}
-;\xi _{Z}^{2}+1 \\ 
\xi _{Z}^{2},\alpha _{Z},\beta _{Z};-
\end{array}
\right. \right) ,  \label{pdfx} \\
F_{\gamma _{Z}}(z)& =\frac{r^{\alpha _{Z}+\beta _{Z}-2}\mathcal{O}_{Z}}{
	\left( 2\pi \right) ^{r-1}}G_{r+1,3r+1}^{3r,1}\left( \frac{\Upsilon
	_{Z}^{r}z }{r^{2r}\mu _{r}^{(Z)}}\left\vert 
\begin{array}{c}
1;\kappa_{1}^{(Z)} \\ 
\kappa_{2}^{(Z)};0%
\end{array}%
\right. \right),  \label{cdfx}
\end{align}
where {$\sqrt[r]{.}$ denotes the $r$th root}, $\mathcal{O}%
_{Z}=\frac{\xi _{Z}^{2}}{\Gamma \left( \alpha _{Z}\right) \Gamma \left(
	\beta _{Z}\right) },$ $\mu _{r}^{(Z)}=\mathbb{E}\left[ \gamma _{Z}\right] $, 
$\Upsilon _{Z}=\frac{\xi _{Z}^{2}\alpha _{Z}\beta _{Z}}{\xi _{Z}^{2}+1},$ $%
\kappa _{1}^{(Z)}=\left\{ \frac{\xi _{Z}^{2}+i}{r}\right\} _{i=1..r},$ $%
\kappa _{2}^{(Z)}=\left\{ \frac{\xi _{Z}^{2}+i}{r},\frac{\alpha _{Z}+i}{r},%
\frac{\beta _{Z}+i}{r}\right\} _{i{\small =0..r-1}},$ and $%
G_{p,q}^{m,n}\left( z\left\vert 
\begin{array}{c}
(a_{i})_{i\leq p} \\ 
(b_{k})_{k\leq q}%
\end{array}%
\right. \right) $ denotes the Meijer G-function \cite[Eq. (9.301)]{Table}

\section{Intercept Probability}

In this section, we derive 
{closed-form and asymptotic
	expressions for the IP of the considered HSTCN}. The IP is defined as the
probability of {the} legitimate link capacity falls below the wiretap
channel one. 
\begin{equation}
\text{IP}=\Pr \left( C_{\sec }\leq 0\right) ,  \label{IP1}
\end{equation}%
where $C_{\sec }$ represents the {system's secrecy capacity} and it can be
expressed as {\ 
	\begin{equation}
	C_{\sec }=\min \left( C_{1S},C_{2S}\right) ,  \label{Cs1}
	\end{equation}%
	with 
	\begin{equation}
	C_{1S}=\log _{2}\left( \frac{1+\gamma _{R}}{1+\gamma _{E_{1}}}\right) ,
	\label{cs1exp1}
	\end{equation}%
	and 
	\begin{equation}
	C_{2S}=\min \left( \log _{2}\left( \frac{1+\gamma _{R}}{1+\gamma _{E_{2}}}%
	\right) ,\log _{2}\left( \frac{1+\gamma _{D}}{1+\gamma _{E_{2}}}\right)
	\right) .  \label{c2s2xp2}
	\end{equation}%
}

\begin{remark}
	One can see from (\ref{IP1}) that the system's secrecy can be enhanced by
	increasing $C_{\sec }$. This can be achieved by either increasing the SNRs
	at legitimate nodes {(i.e., $\mathbf{R} $ and $\mathbf{D} $)} or decreasing
	the SNRs at {the} eavesdroppers at both hops. {To this end, the system's
		secrecy} is impacted by {two} main factors: {(i)} transmit power of the
	sources, and {(ii) fading severity exhibited by different channels}. Owing
	to that, it can be clearly noticed from (\ref{GammaR}) and {(\ref{gammaE1})}
	that increasing either $\overline{\gamma}_{I} $ or $\overline{\gamma}_{S}$
	enhances the SNR at $\mathbf{R}$, while increasing $\overline{\gamma}
	_{S_{J}} $ decreases the SNR at $\mathbf{E_{1}}$. Particularly, the presence
	of a friendly jammer decreases this latter metric as can be ascertained in (%
	\ref{gammaE1}) and (\ref{gammaE1_without_jammer}). Furthermore, above a
	certain threshold of either $\overline{\gamma}_{I}$ or $\overline{\gamma}%
	_{S} $, both legitimate and eavesdropper SNRs depend exclusively of either $%
	\overline{\gamma}_{S}$ or $\overline{\gamma}_{I}$, respectively as can be
	observed in (\ref{gammaE1})-(\ref{UQ}). Consequently, the IP remains steady
	in both aforementioned cases. Likewise, it can be noticed from (\ref{gammaZ}%
	) that the capacities of the second-hop channels are affected by various
	parameters including the satellite's transmit power and the average powers
	of both LOS and multipath components of the downlink turbulence channel.
\end{remark}

\subsection{New framework for the \emph{IP}}

In order to derive the closed-form and asymptotic expressions for the IP of
the considered HSTCN, we have to provide first a framework for IP's
evaluation of a dual-hop cognitive system in the presence of an
eavesdropper at each hop when the relay performs the DF protocol. Next, both
closed-form and asymptotic expressions of {the} IP for the
considered system are provided under two scenarios, namely (i) absence, and
(ii) presence of a friendly jamming signal.

\begin{lemma}
	For a dual-hop cognitive network-aided DF relaying protocol {experiencing
		generalized fading models {over} the intercepting attempt of
		two} wiretappers $\mathbf{E_{1}}$ and $\mathbf{E_{2}}$ at the first and the
	second hop, respectively, {the} \emph{IP} can be evaluated as 
	\begin{equation}
	\emph{IP}=1-\int_{u=0}^{\infty }f_{g_{SP}}\left( u\right) \mathcal{R}
	_{1}\left( u\right) du,  \label{IP_new_expression}
	\end{equation}%
	where 
	\begin{equation}
	\mathcal{R}_{1}\left( u\right) =\int_{y=\gamma _{th}}^{\infty }\mathcal{J}%
	_{1}\left( y,u\right) \mathcal{J}_{2}\left( y\right) dy,  \label{U1}
	\end{equation}%
	\begin{equation}
	\mathcal{J}_{1}\left( y,u\right) =f_{\gamma _{R}|g_{SP}=u}\left( y\right)
	F_{\gamma _{E_{1}}|g_{SP}=u}\left( y\right) ,  \label{J1}
	\end{equation}%
	and 
	\begin{equation}
	\mathcal{J}_{2}\left( y\right) =\int_{z=0}^{y}f_{\gamma _{E_{2}}}\left(
	z\right) F_{\gamma _{D}}^{c}\left( z\right) dz,  \label{J2}
	\end{equation}
	{where $ F^{c}_{.}(.) $ denotes the complementary CDF.} \label%
	{lemma1}
\end{lemma}

\begin{proof}
	The proof is provided in Appendix A. \qedhere
\end{proof}

\subsection{IP closed-form}

\begin{theorem}
	The closed-form expressions for IP of the considered HSTCN {under the
		considered fading models is given} by (\ref{IP_final}) as shown at the top
	of the next page, {for both absence (i.e. $\mathcal{E}=A $) and presence
		(i.e. $\mathcal{E}=P $) of a friendly jamming}, where \label{theorem1}
\end{theorem}

\begin{equation}
e^{\left( Z,\tau \right) }=\frac{\Gamma \left( \alpha _{Z}-\xi
	_{Z}^{2}\right) \Gamma \left( \beta _{Z}-\xi _{Z}^{2}\right) }{\tau +\xi
	_{Z}^{2}},
\end{equation}%
\begin{align}
e^{(Z,\tau ,k)}\left( x,y\right) & =\frac{\left( -1\right) ^{k}\Gamma \left(
	x-y-k\right) }{k!\left( \xi _{Z}^{2}-y-k\right) \left( \tau +y+k\right) }, \\
\xi _{Z}^{2}& \neq y+k,  \notag
\end{align}
\begin{equation}
\mathcal{B}^{(n_{1},n_{2},n_{3},p)}=\frac{\binom{n_{2}}{p}%
	\phi_{S_{J}E_{1}}^{(n_{1})}\phi _{SE_{1}}^{(n_{2})}\phi _{SR}^{(n_{3})}}{
	\upsilon _{S_{J}E_{1}}^{n_{1}+1}\upsilon _{SE_{1}}^{p}},
\end{equation}

\begin{align}
\Psi _{2}^{\left( n_{1,}n_{2},n_{3},p,a,\tau \right) }& =\frac{\varrho
	_{D}^{\tau }\left( \Upsilon _{E_{2}}\varrho _{E_{2}}\right) ^{a}}{\zeta
	^{n_{2}+n_{3}-p+1+a+\tau }}  \label{Psi2} \\
& \times \left[ 
\begin{array}{c}
F_{g_{SP}}\left( \sigma _{S}\right) \mathcal{M}^{\left(
	n_{1,}n_{2},n_{3},p,a,\tau \right) } \\ 
+\lambda _{SP}\sigma _{S}^{a+\tau }\mathcal{Y}^{(n_{1,}n_{2},n_{3},p,a,\tau
	)}%
\end{array}
\right],  \notag
\end{align}

\begin{align}
\Psi _{3}^{\left( n_{1,}n_{2},n_{3},p,a,\tau \right) }& =\frac{\varrho
	_{D}^{\tau }\left( \Upsilon _{E_{2}}\varrho _{E_{2}}\right) ^{a}}{\zeta
	^{n_{2}+n_{3}-p+1+a+\tau }}  \label{Psi3} \\
& \times \left[ 
\begin{array}{c}
F_{g_{SP}}\left( \sigma _{S}\right) \Phi ^{(n_{1,}n_{2},n_{3},p,a,\tau )} \\ 
+\lambda _{SP}\sigma _{S}^{a+\tau }\mathcal{W}^{(n_{1,}n_{2},n_{3},p,a,\tau
	)}%
\end{array}%
\right] ,  \notag
\end{align}%
\begin{equation}
{\theta =\mathcal{O}_{D}\times \mathcal{O}_{E_{2}}},
\end{equation}%
\begin{equation}
\varrho _{Z}=\frac{\overline{\gamma }_{S}}{\mu _{1}^{(Z)}},
\end{equation}%
\begin{equation}
\sigma _{Q}=\frac{\overline{\gamma }_{I}}{\overline{\gamma }_{Q}},
\end{equation}%
\begin{equation}
\epsilon _{I}=\frac{\gamma _{th}}{\overline{\gamma }_{I}},
\end{equation}%
\begin{equation}
\chi =\frac{\upsilon _{S_{J}E_{1}}\zeta }{\upsilon _{SE_{1}}},
\end{equation}%
\begin{equation}
\zeta =\upsilon _{SR}+\upsilon _{SE_{1}},
\end{equation}%
$\tau \in \{0,q\},$ $q\in \{\xi _{D}^{2},\alpha _{D}+k,\beta _{D}+k\},$ {$%
	G_{p,q}^{m,n}\left( z\left\vert 
	\begin{array}{c}
	(a_{l},b_{l})_{l\leq p} \\ 
	(c_{u},d_{u})_{u\leq q}%
	\end{array}%
	\right. \right) $ accounts for the upper incomplete Meijer's G-function \cite%
	[Eq. (1.1.1)]{kilbas}, $j=\sqrt{-1}$, {and} $\mathcal{C}_{s}$ and $\mathcal{C%
	}_{w}$ are {two} vertical lines of integration chosen so as to separate left
	poles of 
	{the integrand functions in (\ref{Ya_final}) and
		(\ref{Wa_final})}, from the right ones.} 
\begin{figure*}[!htb]
	\setcounter{mytempeqncnt}{\value{equation}} \setcounter{equation}{32}
	\par
	\begin{equation}
	\text{IP}_{\mathcal{E}}=1-\frac{\xi _{E_{2}}^{2}\mathcal{D}_{\mathcal{E}
		}^{\left( 0\right) }}{\Gamma \left( \alpha _{E_{2}}\right) \Gamma \left(
		\beta _{E_{2}}\right) }+\theta \left[ e^{(D,0)}\mathcal{D}_{\mathcal{E}
	}^{\left( \xi _{D}^{2}\right) }+\sum_{k=0}^{\infty }e^{(D,0,k)}\left( \beta
	_{D},\alpha _{D}\right) \mathcal{D}_{\mathcal{E}}^{\left( \alpha
		_{D}+k\right) }+e^{(D,0,k)}\left( \alpha _{D},\beta _{D}\right) \mathcal{D}%
	_{ \mathcal{E}}^{\left( \beta _{D}+k\right) }\right] ,\text{ }\mathcal{E\in }
	\left\{ A,P\right\}  \label{IP_final}
	\end{equation}
	\hrulefill{} \vspace*{4pt}
\end{figure*}
%%%%%%%%%%%%%%%%%%%%%%%%%%%%%%%%%%%%%%%%%%%%%%%%%%%%%%%%
\begin{figure*}[!htb]
	\setcounter{mytempeqncnt}{\value{equation}} \setcounter{equation}{33}
	\par
	\begin{equation}
	\mathcal{D}_{A}^{\left( \tau \right) }=\Upsilon _{D}^{\tau }\Delta
	_{SE_{1}}\Delta
	_{SR}\sum_{n_{2}=0}^{m_{SE_{1}}-1}\sum_{n_{3}=0}^{m_{SR}-1}{} \frac{\phi
		_{SR}^{(n_{3})}\phi _{SE_{1}}^{(n_{2})}}{\upsilon _{SE_{1}}^{n_{2}+1}}\left[ 
	\begin{array}{c}
	e^{(E_{2},\tau )}\mathcal{N}_{1}^{\left( n_{2},n_{3},\xi _{E_{2}}^{2},\tau
		\right) }+\sum_{k=0}^{\infty }e^{(E_{2},\tau ,k)}\left( \beta
	_{E_{2}},\alpha _{E_{2}}\right) \mathcal{N}_{1}^{\left( n_{2},n_{3},\alpha
		_{E_{2}}+k,\tau \right) } \\ 
	+\sum_{k=0}^{\infty }e^{(E_{2},\tau ,k)}\left( \alpha _{E_{2}},\beta
	_{E_{2}}\right) \mathcal{N}_{1}^{\left( n_{2},n_{3},\beta _{E_{2}}+k,\tau
		\right) }%
	\end{array}
	\right] {.}  \label{Dtau_no_jammer_final}
	\end{equation}
	\hrulefill{} \vspace*{4pt}
\end{figure*}
%%%%%%%%%%%%%%%%%%%%%%%%%%%%%%
\begin{figure*}[!htb]
	\setcounter{mytempeqncnt}{\value{equation}} \setcounter{equation}{34}
	\par
	\begin{equation}
	\mathcal{D}_{P}^{\left( \tau \right) }=\Upsilon _{D}^{\tau }\Delta _{SR} %
	\left[ 
	\begin{array}{c}
	\sum_{n_{3}=0}^{m_{SR}-1}\phi _{SR}^{(n_{3})}\Xi _{1}^{\left( n_{3},\tau
		\right) }-\Delta _{SE_{1}}\Delta
	_{S_{J}E_{1}}\sum_{n_{1}=0}^{m_{S_{J}E_{1}}-1}\sum_{n_{2}=0}^{m_{SE_{1}}-1}
	\sum_{p=0}^{n_{2}}\sum_{n_{3}=0}^{m_{SR}-1} \\ 
	\times \mathcal{B}^{(n_{1},n_{2},n_{3},p)}\left\{ F_{g_{S_{J}P}}\left(
	\sigma _{S_{J}}\right) \Xi _{2}^{(n_{1,}n_{2},n_{3},p,\tau )}+\Xi
	_{3}^{(n_{1,}n_{2},n_{3},p,\tau )}\right\}%
	\end{array}
	\right] {.}  \label{Dtau_jammer_final}
	\end{equation}
	\hrulefill{} \vspace*{4pt}
\end{figure*}
%%%%%%%%%%%%%%%%%%%%%%%%%
\begin{figure*}[!htb]
	\setcounter{mytempeqncnt}{\value{equation}} \setcounter{equation}{35}
	\par
	\begin{equation}
	\mathcal{N}_{1}^{\left( n_{2},n_{3},a,\tau \right) }=\frac{\varrho
		_{D}^{\tau }\left( \Upsilon _{E_{2}}\varrho _{E_{2}}\right) ^{a}}{\upsilon
		_{SR}^{n_{3}+\tau +a+1}}\left[ 
	\begin{array}{c}
	F_{g_{SP}}\left( \sigma _{S}\right) G_{2,2}^{1,2}\left( \frac{\upsilon
		_{SE_{1}}}{\upsilon _{SR}}\left\vert 
	\begin{array}{c}
	\left( 1,0\right) ,\left( -n_{3}-\tau -a,\epsilon _{I}\sigma _{S}\upsilon
	_{SR}\right) ;- \\ 
	\left( {\small n}_{2}{\small +1,0}\right) ;\left( 0,0\right)%
	\end{array}%
	\right. \right) \\ 
	+\frac{\left( \lambda _{SP}\sigma _{S}\right) ^{\tau +a}}{\left( 2\pi
		j\right) ^{2}}\int_{C_{s}}\frac{\Gamma \left( {\small n}_{2}{\small +1+s}%
		\right) \Gamma \left( -{\small s}\right) }{\Gamma \left( 1-{\small s}\right) 
	}\left( \frac{\upsilon _{SE_{1}}}{\upsilon _{SR}}\right) ^{-s} \\ 
	\times \int_{C_{w}}\frac{\Gamma \left( -\tau -a-w+1,\lambda _{SP}\sigma
		_{S}\right) \Gamma \left( n_{3}+\tau +a+1-s+w\right) \Gamma \left( w\right) 
	}{\Gamma \left( 1+w\right) }\left( \frac{\epsilon _{I}\upsilon _{SR}}{%
		\lambda _{SP}}\right) ^{-w}dsdw%
	\end{array}%
	\right] {.}  \label{N1_final}
	\end{equation}%
	\par
	\hrulefill{} \vspace*{4pt}
\end{figure*}
%%%%%%%%%%%%%%%%%%%%%%%%%
\begin{figure*}[!htb]
	\setcounter{mytempeqncnt}{\value{equation}} \setcounter{equation}{36}
	\par
	\begin{equation}
	\Xi _{1}^{\left( n_{3},\tau \right) }=e^{(E_{2},\tau )}\Psi _{1}^{\left(
		n_{3},\xi _{E_{2}}^{2},\tau \right) }+\sum_{k=0}^{\infty }e^{(E_{2},\tau
		,k)}\left( \beta _{E_{2}},\alpha _{E_{2}}\right) \Psi _{1}^{\left(
		n_{3},\alpha _{E_{2}}+k,\tau \right) }+e^{(E_{2},\tau ,k)}\left( \alpha
	_{E_{2}},\beta _{E_{2}}\right) \Psi _{1}^{\left( n_{3},\beta _{E_{2}}+k,\tau
		\right) }{.}  \label{E1}
	\end{equation}
	\hrulefill{} \vspace*{4pt}
\end{figure*}
%%%%%%%%%%%%%%%%%%%%%%%%%%%%%%%
\begin{figure*}[!htb]
	\setcounter{mytempeqncnt}{\value{equation}} \setcounter{equation}{37}
	\par
	\begin{equation}
	\Psi _{1}^{\left( n_{3},a,\tau \right) }=\frac{\varrho _{D}^{\tau }\left(
		\Upsilon _{E_{2}}\varrho _{E_{2}}\right) ^{a}}{\upsilon _{SR}^{n_{3}+\tau
			+a+1}}\left[ 
	\begin{array}{c}
	F_{g_{SP}}\left( \sigma _{S}\right) \Gamma \left( n_{3}+\tau +a+1,\epsilon
	_{I}\sigma _{S}\upsilon _{SR}\right) \\ 
	+\left( \sigma _{S}\lambda _{SP}\right) ^{\tau +a}G_{2,2}^{2,1}\left( \frac{
		\upsilon _{SR}\epsilon _{I}}{\lambda _{SP}}\left\vert 
	\begin{array}{c}
	\left( \tau +a,\sigma _{S}\lambda _{SP}\right) ;\left( 1,0\right) \\ 
	\left( 0,0\right) ,\left( n_{3}+\tau +a+1,0\right) ;-%
	\end{array}
	\right. \right)%
	\end{array}
	\right] {.}  \label{Psi1}
	\end{equation}
	\hrulefill{} \vspace*{4pt}
\end{figure*}
%%%%%%%%%%%%%%%%%%%%%%%%%%%%%%%
\begin{figure*}[!htb]
	\setcounter{mytempeqncnt}{\value{equation}} \setcounter{equation}{38}
	\par
	\begin{align}
	\Xi _{n}^{(n_{1,}n_{2},n_{3},p,\tau )}& =e^{(E_{2},\tau )}\Psi _{n}^{\left(
		n_{1,}n_{2},n_{3},p,\xi _{E_{2}}^{2},\tau \right) }+\sum_{k=0}^{\infty
	}e^{(E_{2},\tau ,k)}\left( \beta _{E_{2}},\alpha _{E_{2}}\right) \Psi
	_{n}^{\left( n_{1,}n_{2},n_{3},p,\alpha _{E_{2}}+k,\tau \right) }  \notag \\
	& +e^{(E_{2},\tau ,k)}\left( \alpha _{E_{2}},\beta _{E_{2}}\right) \Psi
	_{n}^{\left( n_{1,}n_{2},n_{3},p,\beta _{E_{2}}+k,\tau \right) },\text{ }
	n\in \{2,3\}{.}  \label{Em}
	\end{align}
	\hrulefill{} \vspace*{4pt}
\end{figure*}
\begin{figure*}[!htb]
	\setcounter{mytempeqncnt}{\value{equation}} \setcounter{equation}{39}
	\par
	\begin{equation}
	\mathcal{M}^{(n_{1,}n_{2},n_{3},p,a,\tau )}=G_{2,3}^{2,2}\left( \frac{\chi
		\sigma _{S_{J}}}{\overline{\gamma }_{I}}\left\vert 
	\begin{array}{c}
	\left( -p,0\right) ,\left( 1,0\right) ;- \\ 
	\left( n_{1}+1,0\right) ,\left( n_{2}+n_{3}-p+1+\tau +a,\zeta \epsilon
	_{I}\sigma _{S}\right) ;\left( 0,0\right)%
	\end{array}
	\right. \right) {.}  \label{M2}
	\end{equation}
	\hrulefill{} \vspace*{4pt}
\end{figure*}
\begin{figure*}[!htb]
	\setcounter{mytempeqncnt}{\value{equation}} \setcounter{equation}{40}
	\par
	\begin{align}
	\mathcal{Y}^{(n_{1,}n_{2},n_{3},p,a,\tau )}& =\frac{-\lambda _{SP}^{a+\tau
			-1}}{\left( 2\pi j\right) ^{2}}\int_{\mathcal{C}_{s}}\frac{\Gamma \left(
		n_{1}+1+s\right) \Gamma \left( 1+p-s\right) }{s}\left( \frac{\chi \sigma
		_{S_{J}}}{\overline{\gamma }_{I}}\right) ^{-s}  \label{Ya_final} \\
	& \times \int_{\mathcal{C}_{w}}\frac{\Gamma \left( n_{2}+n_{3}-p+1+\tau
		+a+s+w\right) \Gamma \left( 1-a-\tau -w,\sigma _{S}\lambda _{SP}\right) }{w}
	\left( \frac{\zeta \epsilon _{I}}{\lambda _{SP}}\right) ^{-w}dsdw{.}  \notag
	\end{align}
	\hrulefill{} \vspace*{4pt}
\end{figure*}
%%%%%%%%%%%%%%%%%%%%%%%%%%%%%%%%%%%%%%%%%%%%%
\begin{figure*}[!htb]
	\setcounter{mytempeqncnt}{\value{equation}} \setcounter{equation}{41}
	\par
	\begin{equation}
	\Phi ^{(n_{1,}n_{2},n_{3},p,a,\tau )}=G_{3,3}^{2,3}\left( \frac{\chi }{
		\lambda _{S_{J}P}\overline{\gamma }_{I}}\left\vert 
	\begin{array}{c}
	\left( -p,0\right) ,\left( 1,0\right) ,\left( 0,\sigma _{S_{J}}\lambda
	_{S_{J}P}\right) ;- \\ 
	\left( n_{1}+1,0\right) ,\left( n_{2}+n_{3}-p+1+\tau +a,\zeta \epsilon
	_{I}\sigma _{S}\right) ;\left( 0,0\right)%
	\end{array}
	\right. \right) {.}  \label{Phi2}
	\end{equation}%
	\par
	\hrulefill{} \vspace*{4pt}
\end{figure*}
\begin{figure*}[!htb]
	\setcounter{mytempeqncnt}{\value{equation}} \setcounter{equation}{42}
	\par
	\begin{align}
	\mathcal{W}^{(n_{1,}n_{2},n_{3},p,a,\tau )}& =\frac{-\lambda _{SP}^{a+\tau
			-1}}{\left( 2\pi j\right) ^{2}}\int_{\mathcal{C}_{s}}\frac{\Gamma \left(
		n_{1}+1+s\right) \Gamma \left( 1+p-s\right) \Gamma \left( 1-s,\sigma
		_{S_{J}}\lambda _{S_{J}P}\right) }{s}\left( \frac{\chi }{\lambda _{S_{J}P} 
		\overline{\gamma }_{I}}\right) ^{-s}  \label{Wa_final} \\
	& \times \int_{\mathcal{C}_{w}}\frac{\Gamma \left( n_{2}+n_{3}-p+1+a+\tau
		+s+w\right) \Gamma \left( 1-a-\tau -w,\sigma _{S}\lambda _{SP}\right) }{w}
	\left( \frac{\zeta \epsilon _{I}}{\lambda _{SP}}\right) ^{-w}dsdw.  \notag
	\end{align}%
	\par
	\hrulefill{} \vspace*{4pt}
\end{figure*}

\begin{proof}
	The proof is provided in Appendix B. \qedhere
\end{proof}

\subsection{Asymptotic IP}

In this subsection, we provide an asymptotic analysis of the derived
closed-form expression for the IP {in high SNR regime}. 
{It can be noticed from (\ref{M2}), (\ref{Phi2}),
	(\ref{Ya_final}), and (\ref{Wa_final}) that the expression for the IP can be
	approximated for high SNR values by considering
	$\overline{\gamma}_{I}\rightarrow \infty.$}

\begin{theorem}
	The Asymptotic expression for the \emph{IP} in the presence of a friendly jammer is given by (\ref{IP_asymptotic}) as shown at the top of page {7}, with \label{theorem2}
\end{theorem}

\begin{equation}
\Psi _{2}^{\left( 0,n_{2},n_{3},p,a,\tau \right) }= \sigma _{S_{J}}\mathcal{V%
}^{\left( n_{2},n_{3},p,a,\tau \right) },  \label{Psi2_approximate}
\end{equation}

\begin{equation}
\Psi _{3}^{\left( 0,n_{2},n_{3},p,a,\tau \right) }= \frac{\Gamma \left(
	2+n_{1},\sigma _{S_{J}}\lambda _{S_{J}P}\right) }{\lambda _{S_{J}P}}\mathcal{%
	\ V}^{\left( n_{2},n_{3},p,a,\tau \right) },  \label{Psi3_approximate}
\end{equation}
{and $\mathcal{V}^{\left( n_{2},n_{3},p,a,\tau \right)}$ is
	given by (\ref{V_asymp}) as shown at the same aforementioned page.} 
\begin{figure*}[!htb]
	\setcounter{mytempeqncnt}{\value{equation}} \setcounter{equation}{45}
	\par
	\begin{equation}
	\text{IP}_{P}^{\left( \infty \right) }\sim 1-\frac{\xi _{E_{2}}^{2}\mathcal{D%
		}_{P}^{\left( 0,\infty \right) }}{\Gamma \left( \alpha _{E_{2}}\right)
		\Gamma \left( \beta _{E_{2}}\right) }+\theta \left[e^{(D,0)}\mathcal{D}%
	_{P}^{\left( \xi _{D}^{2},\infty \right) }+\sum_{k=0}^{\infty
	}e^{(D,0,k)}\left( \beta _{D},\alpha _{D}\right) \mathcal{D}_{P}^{\left(
		\alpha _{D}+i,\infty \right) }+e^{(D,0,k)}\left( \alpha _{D},\beta
	_{D}\right) \mathcal{D}_{P}^{\left( \beta _{D}+k,\infty \right) }\right].
	\label{IP_asymptotic}
	\end{equation}
	\hrulefill{} \vspace*{4pt}
\end{figure*}
%%%%%%%%%%%%%%%%%%%%%%
\begin{figure*}[!htb]
	\setcounter{mytempeqncnt}{\value{equation}} \setcounter{equation}{46}
	\par
	\begin{equation}
	\mathcal{D}_{P}^{\left( \tau ,\infty \right) }= \Upsilon _{D}^{\tau }\Delta
	_{SR}\left[ 
	\begin{array}{c}
	\sum_{n_{3}=0}^{m_{SR}-1}\phi _{SR}^{(n_{3})}\Xi _{1}^{\left( n_{3},\tau
		\right) }-\Delta _{SE_{1}}\Delta
	_{S_{J}E_{1}}\sum_{n_{2}=0}^{m_{SE_{1}}-1}\sum_{p=0}^{n_{2}}%
	\sum_{n_{3}=0}^{m_{SR}-1} \\ 
	\times \left\{ F_{g_{S_{J}P}}\left( \sigma _{S_{J}}\right) \Xi
	_{2}^{(0,n_{2},n_{3},p,\tau )}+\Xi _{3}^{(0,n_{2},n_{3},p,\tau )}\right\} 
	\mathcal{B}^{(n_{1},n_{2},n_{3},p)}%
	\end{array}
	\right] {.}
	\end{equation}
	\hrulefill{} \vspace*{4pt}
\end{figure*}
%%%%%%%%%%%%%%%%%%%%%%%%%%%
\begin{figure*}[!htb]
	\setcounter{mytempeqncnt}{\value{equation}} \setcounter{equation}{47}
	\par
	\begin{align}
	\Xi _{n}^{(0,n_{2},n_{3},p,\tau )}& = e^{(E_{2},\tau )}\Psi _{n}^{\left(
		0,n_{2},n_{3},p,\xi _{E_{2}}^{2},\tau \right) }+\sum_{k=0}^{\infty
	}e^{(E_{2},\tau ,k)}\left( \beta _{E_{2}},\alpha _{E_{2}}\right) \Psi
	_{n}^{\left( 0,n_{2},n_{3},p,\alpha _{E_{2}}+k,\tau \right) }  \notag \\
	& +e^{(E_{2},\tau ,k)}\left( \alpha _{E_{2}},\beta _{E_{2}}\right) \Psi
	_{n}^{\left( 0,n_{2},n_{3},p,\beta _{E_{2}}+k,\tau \right) },\text{ }n\in
	\{2,3\}{.}
	\end{align}
	\hrulefill{} \vspace*{4pt}
\end{figure*}
\begin{figure*}[!htb]
	\setcounter{mytempeqncnt}{\value{equation}} \setcounter{equation}{48}
	\par
	\begin{equation}
	\mathcal{V}^{\left( n_{2},n_{3},p,a,\tau \right) }=\frac{\chi \Gamma \left(
		2+p\right) }{\overline{\gamma }_{I}}\left[ 
	\begin{array}{c}
	F_{g_{SP}}\left( \sigma _{S}\right) \Gamma \left( n_{2}+n_{3}-p+\tau
	+a,\zeta \epsilon _{I}\sigma _{S}\right) \\ 
	+\left( \lambda _{SP}\sigma _{S}\right) ^{a+\tau }{\small G}%
	_{2,2}^{2,1}\left( \frac{{\small \zeta \epsilon }_{I}}{{\small \lambda }_{SP}%
	}\left\vert 
	\begin{array}{c}
	\left( {\small a+\tau ,\sigma }_{S}{\small \lambda }_{SP}\right) ;\left( 
	{\small 1,0}\right) \\ 
	\left( 0,0\right) ,\left( {\small n}_{2}{\small +n}_{3}{\small -p+\tau +a,0}%
	\right) ;-%
	\end{array}%
	\right. \right)%
	\end{array}%
	\right] .  \label{V_asymp}
	\end{equation}
	\hrulefill{} \vspace*{4pt}
\end{figure*}

\begin{proof}
	The proof is provided in Appendix C. \qedhere
\end{proof}

\begin{remark}
	It is worth mentioning that the expression for the IP in the absence of a
	friendly jammer does not have an asymptotic expression as (\ref%
	{Dtau_no_jammer_final}) is independent of $\overline{\gamma}_{I}$.
\end{remark}

\section{Numerical Results and Discussion}
In this {part}, the derived analytical results are validated through Monte Carlo simulation by generating $10^{6}$ random samples and setting {the} parameters are {summarized} in Table \ref{table1}. The turbulence parameters of the FSO {hops} were generated based on OGS-satellite distance, wavelength, and aperture radius according to \cite[Eqs. (4, 9-10)]{Koushal} and \cite[Eqs. (8)]{Sandalidis}. {The main point of note from Figures 2-7 is that all} closed-form and simulation curves are perfectly {matching} for {numerous system} parameters' values, {showing the high accuracy of our results}. 

\begin{table}[tbp]
	\begin{center}
		\caption{Simulation parameters.}
		\begin{tabular}{c|c|c|c|c|c}
			\hline
			\hline
			\textit{Parameter} &$ b_{X} $&$ m_{X} $& $ \Omega_{X} $ & $\lambda_{QP}$ & $ \alpha_{Z} $  \tabularnewline
			\hline
			\textit{Value}  & 1.4 & 2 & 3 & 0.8 & 6.1096   \tabularnewline
			\hline
			\textit{Parameter} & $ \beta_{Z} $ & $ \xi _{Z} $ & $ \gamma_{th}$(dB) & $ \overline{\gamma}_{I}$(dB) &  $ \overline{\gamma}_{S}$(dB)   \tabularnewline
			\hline
			\textit{Value} & 1.0794 & 1.1227 & 2 & 9 &  60 \tabularnewline
			\hline
			\textit{Parameter} &  $ \overline{\gamma}_{S_{J}}$(dB) & $\mu_{E_{2}}$(dB) & $\mu_{D}$(dB) & $ \eta $ & $ \omega_{D} $   \tabularnewline
			\hline
			\textit{Value} & 10 & 20 & 40 & 0.7 &  0.7  \tabularnewline
			\hline
			\hline
		\end{tabular}
		\label{table1}
	\end{center}
\end{table}

Fig. 2 depicts the IP versus $\overline{\gamma }_{I}$ for
various values of $\Omega _{X}$. It is clearly seen that the greater $\overline{\gamma }_{I}$ {is}, the smaller the IP {is}. This can be justified from (\ref{GammaR}) by the fact that when the {MTIP} at the PU receiver increases, the SU is allowed to use its maximal transmit power, which contributes to the improvement of the {SS}.

Figures 3 and 4 show the IP versus $\overline{\gamma }_{I} $ and $\overline{\gamma}_{S}$, respectively, for various values of $\overline{\gamma}_{S_{J}}$.  {It can be ascertained that} the IP decreases with the increase of $\overline{\gamma }_{I}$, $\overline{\gamma }_{S}$, and $\overline{\gamma }_{S_{J}}$ as explained in Remark 1. Also, it can be noticed that the presence of a friendly jammer improves the SS per the same Remark. However, one can notice that for low values of $\overline{\gamma }_{I}$ and $\overline{\gamma}_{S}$, the friendly jammer does not contribute to the enhancement of the SS. In fact, it can be seen from (\ref{UQ}) that the smaller $\overline{\gamma }_{I}$  and $\overline{\gamma}_{S}$ are, the smaller $U_{S}$ is. Thus, it follows from (\ref{gammaE1_without_jammer}) and (\ref{gammaE1}) that both $ \gamma_{E_{1}} $ and $ \gamma^{(J)}_{E_{1}} $ approach 0. Moreover, it can be observed that above certain thresholds of either $\overline{\gamma }_{I}$ or $\overline{\gamma }_{S}$, respectively, the IP becomes steady as discussed in Remark 1.  

%%%%%%%%%%%%%%%%%%%%%%%%%%%%%%%%%%%%%%%%%%%%%%%%%
Fig. 5 shows the IP versus $\overline{\gamma }_{S_{J}}$ for
various values of $\Omega _{X}$. As can be seen, the IP decreases with
the increasing values of the $\overline{\gamma }_{S_{J}}$. This can be
justified from (\ref{gammaE1}) {as} increasing $\overline{\gamma 
}_{S_{J}}$ decreases the SNR at the eavesdropper which reduces the wiretap link {capacity}. Consequently, the secrecy capacity gets enhanced which
results in an improvement of the {SS}. 

Fig. 6 illustrates the IP as a function of $\mu _{D} $ in the presence and absence of a friendly jammer for various values of $\Omega_{X}$. {The greater $\mu_{D} $ is, the greater the legitimate end-user SNR is, leading to the improvement of the SS.} This behavior can be {interpreted as} increasing $\mu _{D} $ leads to the enhancement of the SNR at the destination which improves the legitimate link capacity accordingly.

Figures 7 and 8 depict the IP as a function of the average power of the LOS and multipath components in the presence and absence of a friendly jammer. These powers are assumed to be identical for all channels i.e., $ \Omega_{SR}=\Omega_{SE_{1}}=\Omega_{S_{J}E_{1}}$, and $ b_{SR}=b_{SE_{1}}=b_{S_{J}E_{1}} $. One can ascertain that increasing the average powers of the LOS and multi-path components at the first hop result in an enhancement of the {SS}. Moreover, it is clearly seen that the presence of a friendly jammer is strengthening the SS. For instance, one can see that for $\Omega_{X}=6$ and  $b_{X}=6$, {IP equals $0.35$ and $0.61$} in the presence and absence of a friendly jammer, respectively.
\begin{table*}[tbp]
	%\vspace{-.675em} \hspace{-.675em}
	\begin{tabular}{p{7cm}p{10cm}}
		\hspace*{-.54cm}\includegraphics[scale=0.45]{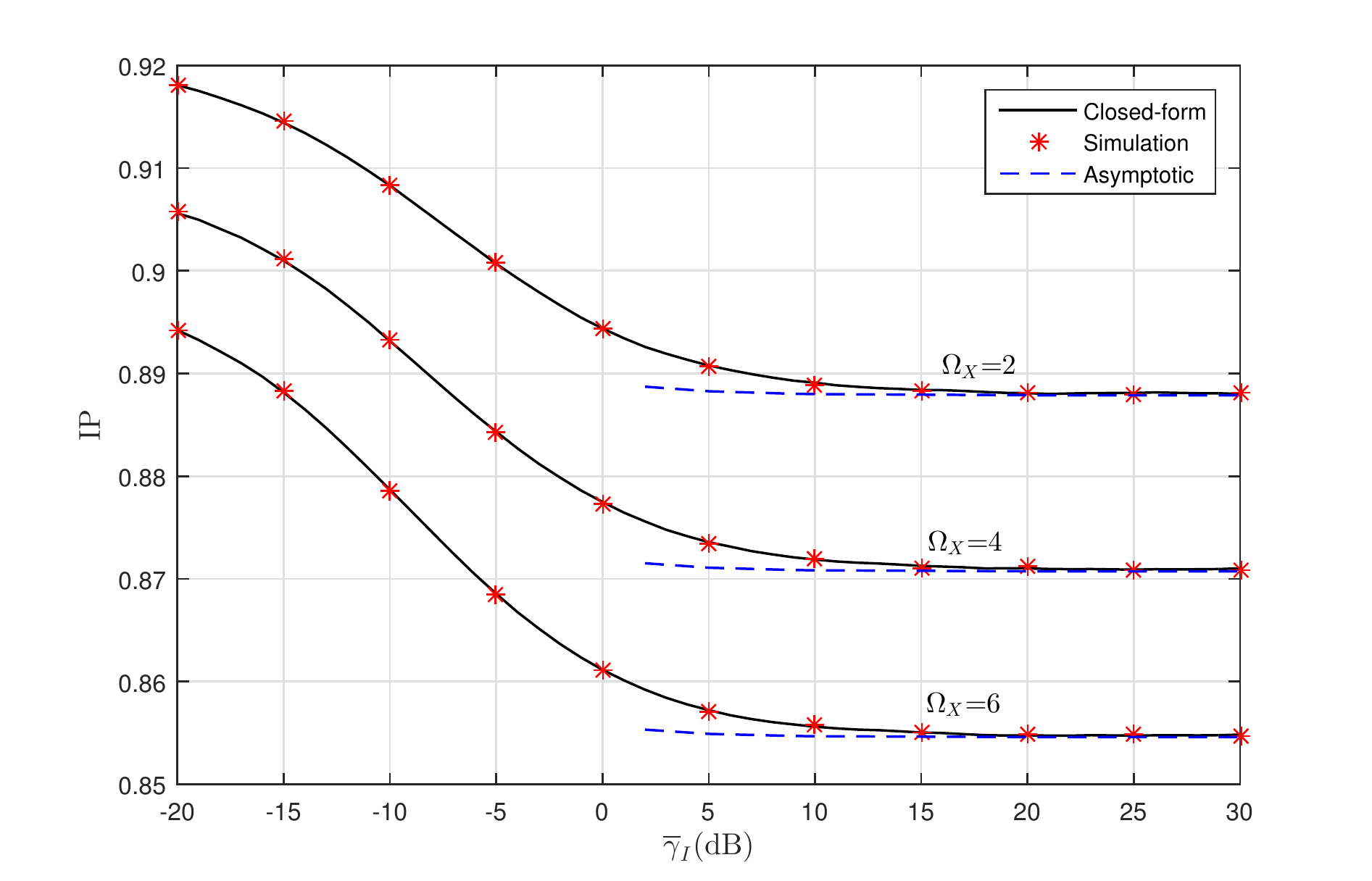} & \hspace*{2cm}\includegraphics[scale=0.45]{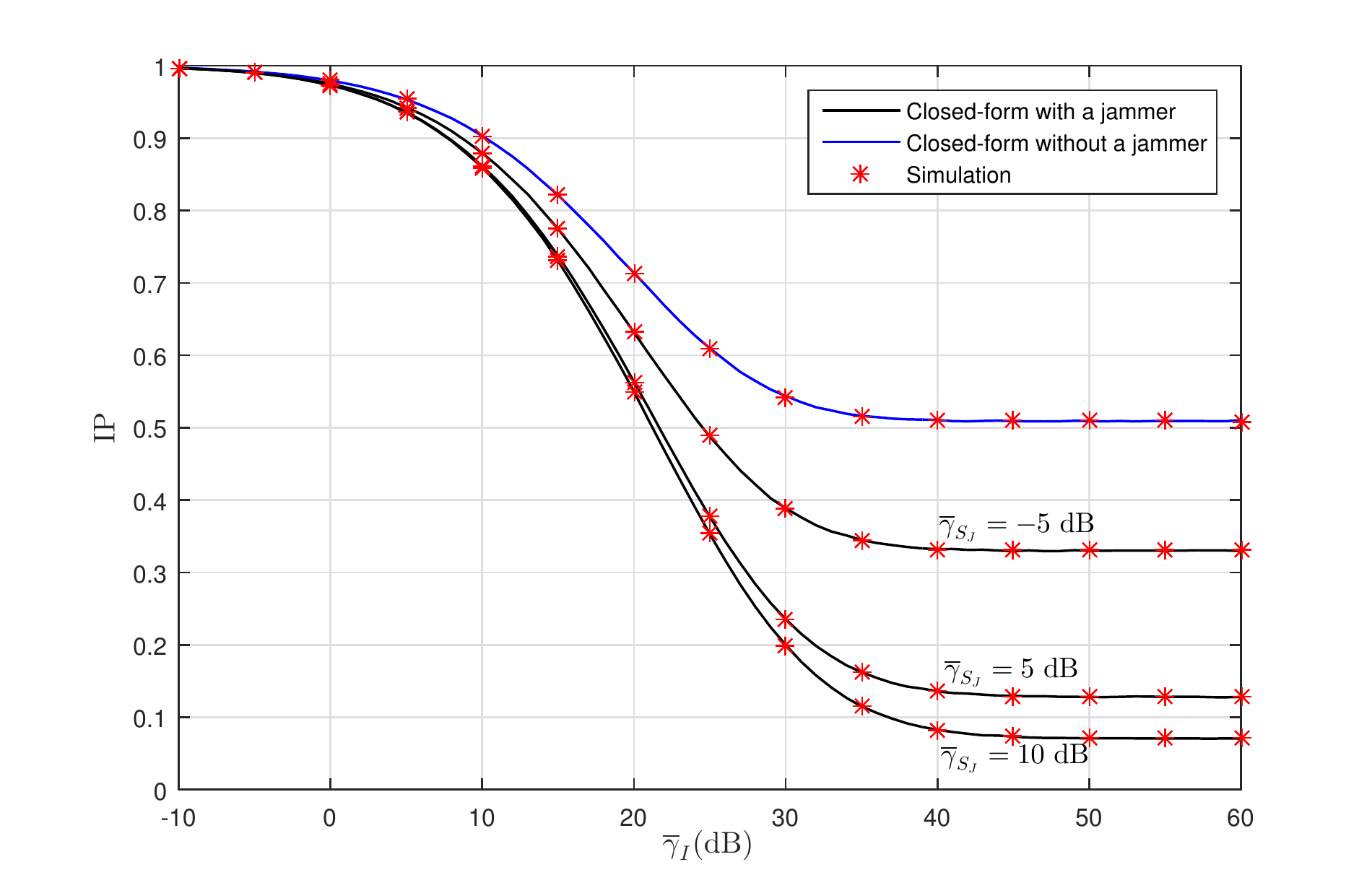}  \\
		{\footnotesize Fig. 2. IP vs $\overline{\protect\gamma }_{I}$ in the presence of a
			friendly jammer for different values of $\Omega_{X}$, $\protect\rho %
			_{D}=0.001 $, $\protect\rho_{E_{2}}=0.01$, $\protect\sigma _{S}=\protect%
			\sigma_{S_{J}}=1,$ $\protect\epsilon _{I}=0.1,$ and $b_{X}=4$.} & \hspace*{4cm}{\footnotesize Fig. 3.  IP vs $\overline{\protect\gamma }_{I}$ for different values of $\overline{\protect\gamma }_{S_{J}}$.} \\
		%%%%%%%%%%%%%%%%%%%%%%%%%%%%%%%%%%%%%%%%%%%%%%%%%%%%%%%%%%%%%%%%%%%%%%%%%%%%%%%%%%%%%%%%%%%%%%%%%%%%%%%%%%%%%%%
		\hspace*{-.54cm}\includegraphics[scale=0.45]{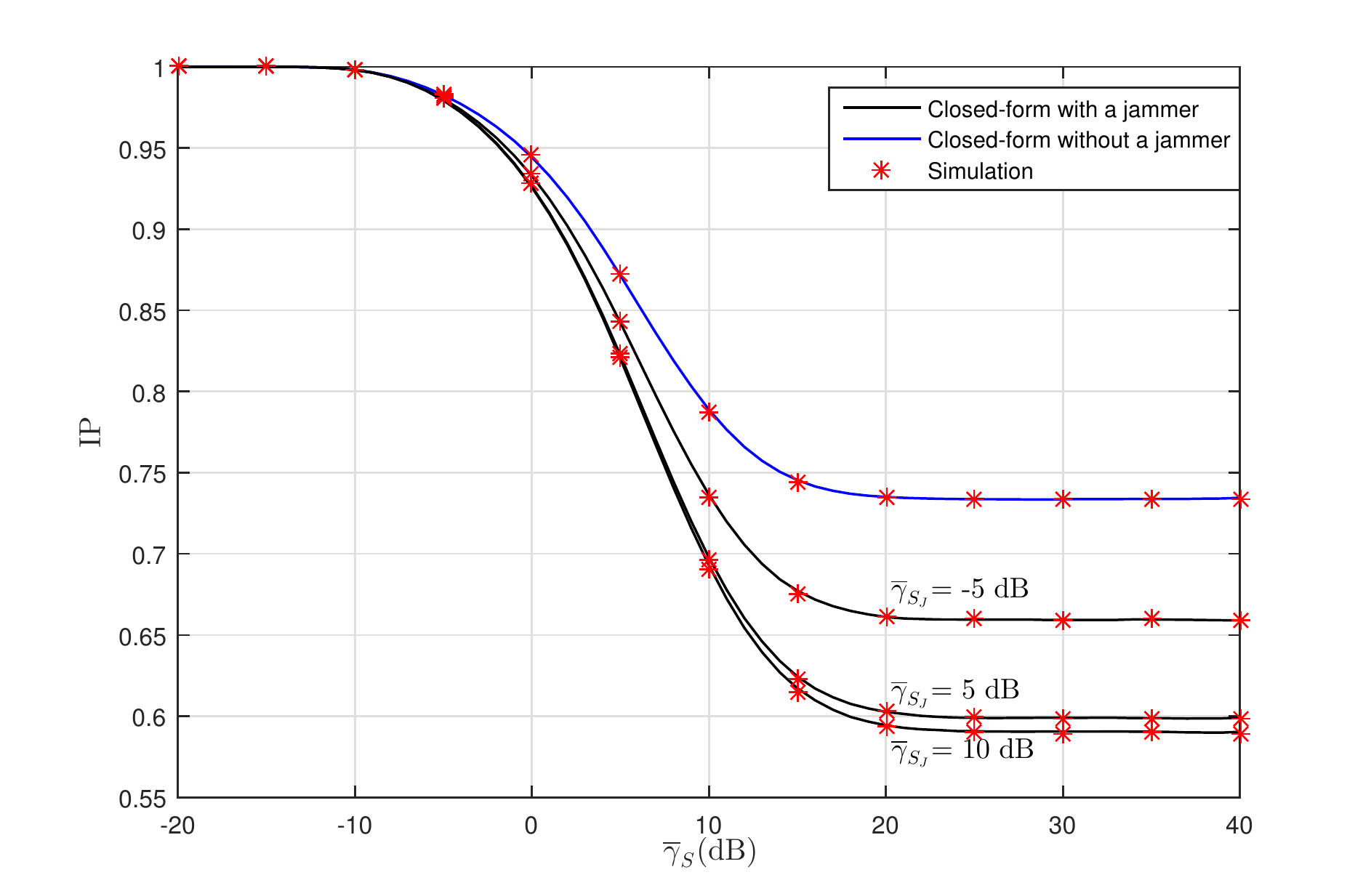} & \hspace*{2cm}\includegraphics[scale=0.45]{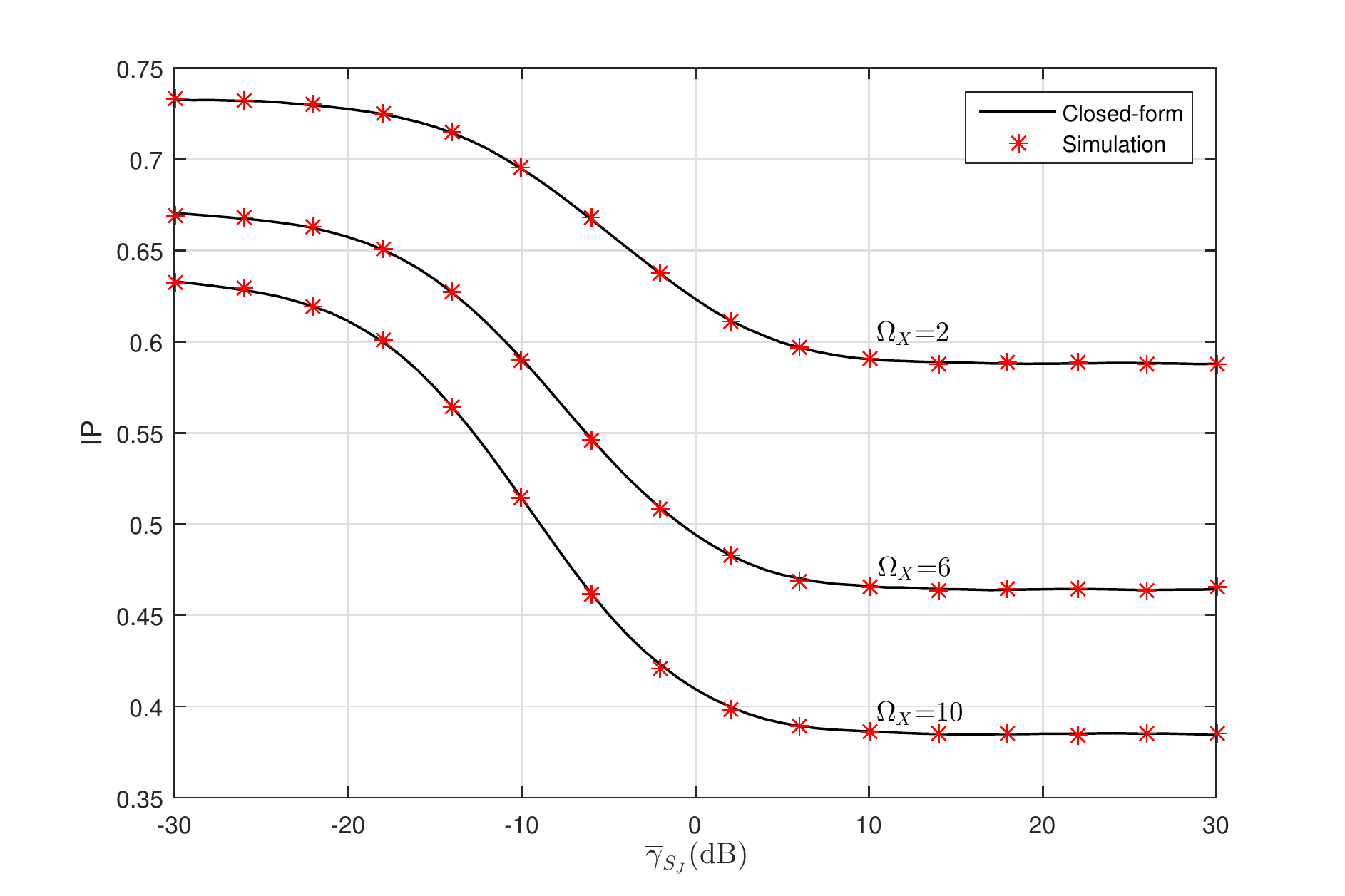}  \\
		\hspace*{1.5cm}{\footnotesize Fig. 4. IP vs $\overline{\protect\gamma}_{S}$ for different values of $\overline{\protect\gamma }_{S_{J}}$.} & \hspace*{3.75cm}{\footnotesize Fig. 5. IP vs $\overline{\protect\gamma }_{S_{J}}$ for different values of $\Omega _{X}$.} \\
		%%%%%%%%%%%%%%%%%%%%%%%%%%%%%%%%%%%%%%%%%%%%%%%%%%%%%%%%%%%%%%%%%%%%%%%%%%%%%%%%%%%%%%%%%%%%%%%%%%%%%%%%%%%%%%%	
		\hspace*{-.54cm}\includegraphics[scale=0.45]{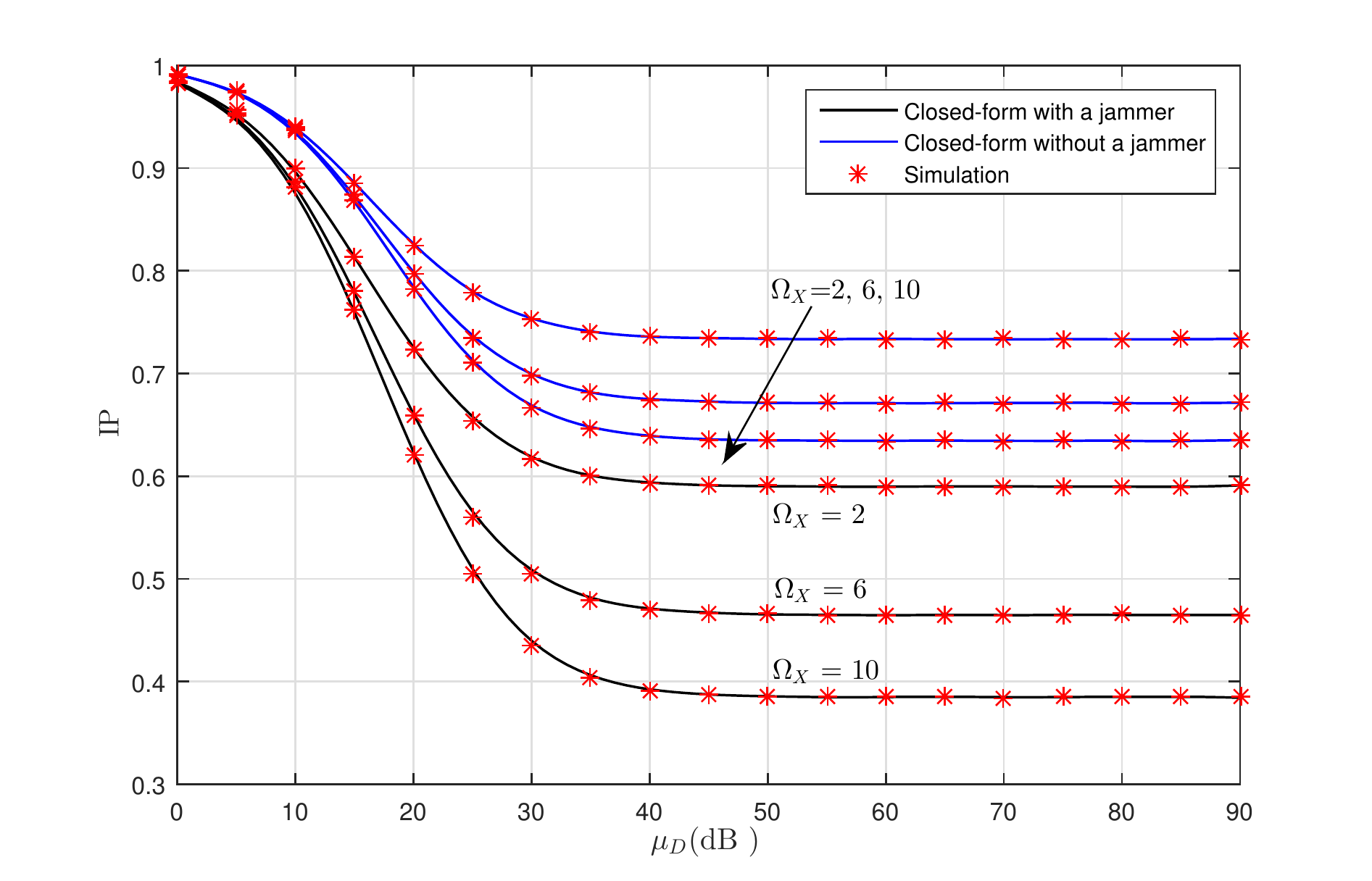} & \hspace*{2cm}\includegraphics[scale=0.5]{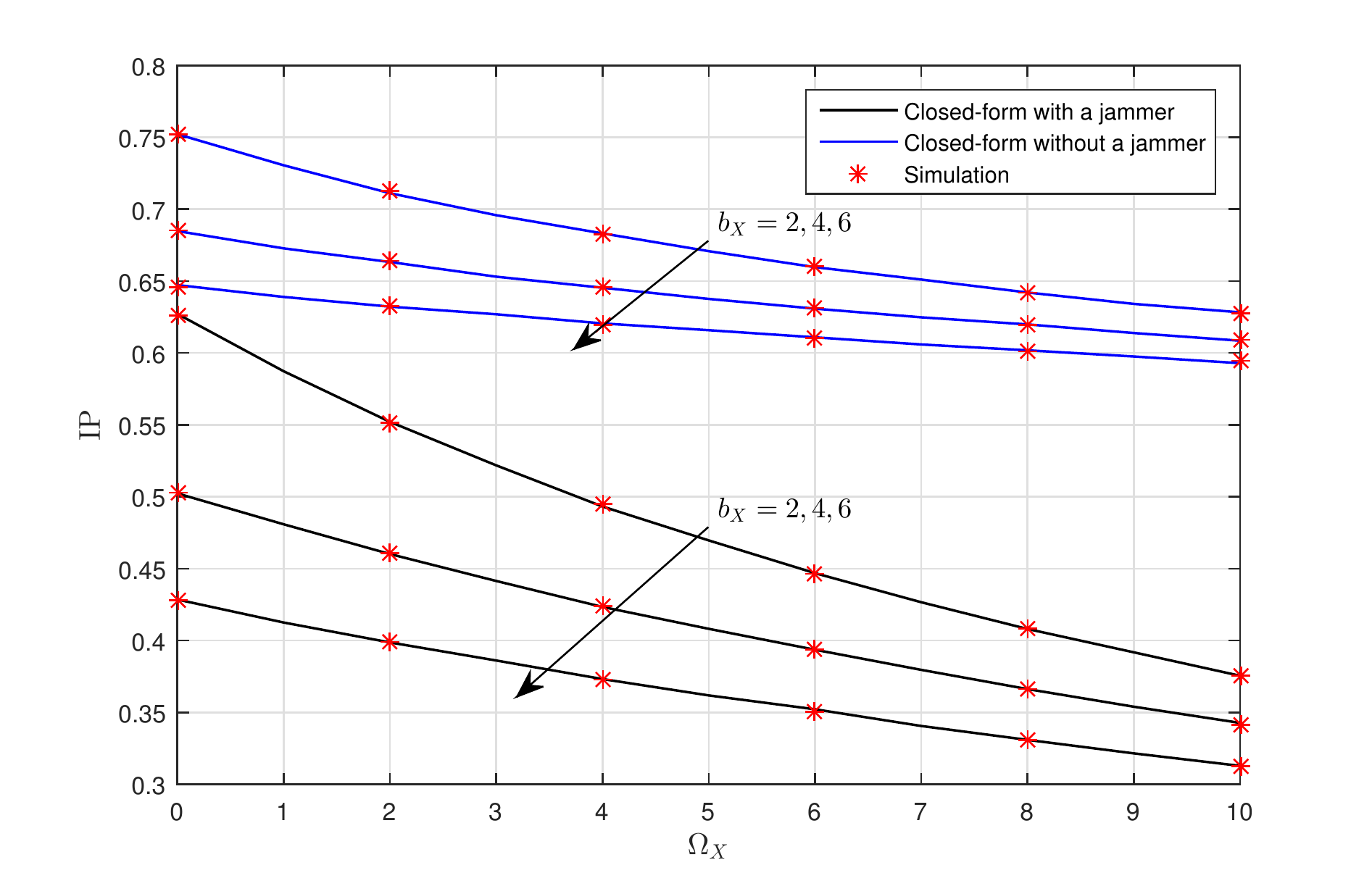}  \\
		\hspace*{1.5cm}{\footnotesize Fig. 6. IP vs $\protect\mu _{D}$ for different values of $\Omega _{X}$.} & \hspace*{3.75cm}{\footnotesize Fig. 7.  IP vs $\Omega _{X}$ for different values of $b_{X}$.}
	\end{tabular}
\end{table*}
%%%%%%%%%%%%%%%%%%%%%%%%%%%%%%%%%%%%%%%%%%%%%%%%%%%%%%%%%%%%%%%%%%%%%%%%%%%%%%%%%%%%%%%%%%%%%%%%%%%%%%%%%%%%%%%%%
\begin{figure}[tbp]
	\begin{centering}
		\includegraphics[scale=0.45]{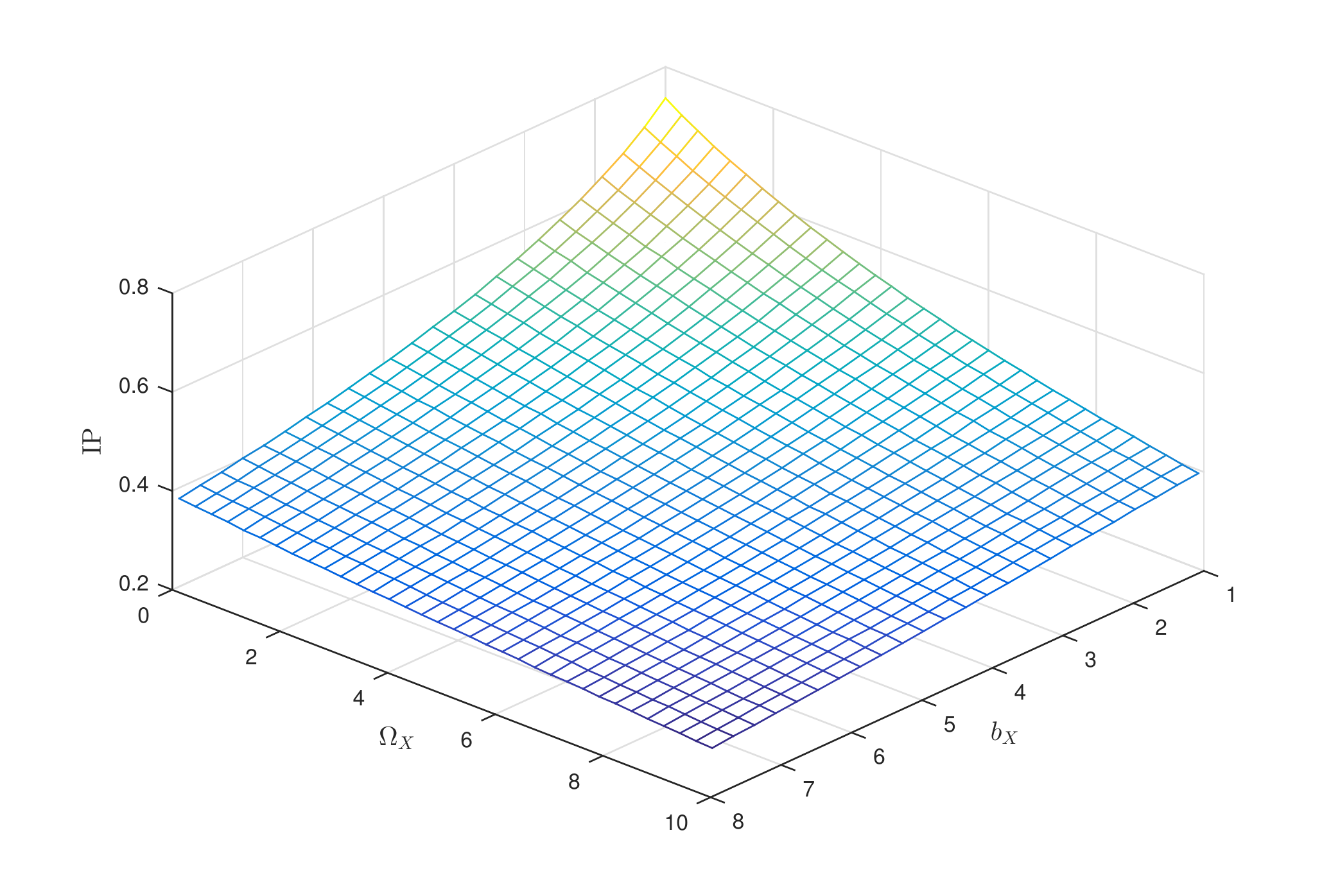}
	\end{centering}
	\centering{}
	\caption*{{\footnotesize Fig. 8. IP vs $\Omega_{X}$ and $b_{X}$ in the presence of a friendly jammer.}}
\end{figure}
\section{Conclusion}

In this paper, the physical layer security of a hybrid satellite-terrestrial
cognitive network was investigated. Specifically, the main aim of this work
was to derive a new formula for the IP of a DF relaying system that takes
into consideration the presence of two eavesdroppers (i.e. at the first and
second hop). This formula was then used to investigate the secrecy
performance of our system by deriving closed-form and asymptotic
expressions for the IP in the presence and absence of a friendly jammer. By
considering different key parameters of the network, our results
demonstrated that the best secrecy can be achieved by increasing the maximum
tolerated interference power at the PU receiver and the sources transmit
power. Interestingly, we showed that friendly jammer does not contribute
to the enhancement of the {SS} when the transmit power of the source or the
maximum tolerated interference power at the PU is below certain thresholds.

%\appendices
\section*{Appendix A: proof of Lemma 1}
Using {equations (\ref{IP1})-(\ref{c2s2xp2})}, the IP can be %
{rewritten} as 
\begin{align}
\text{IP}& \overset{\left( a\right) }{=}\Pr \left( \left. C_{\sec }\leq
0\right\vert \gamma _{R}>\gamma _{th}\right) \Pr \left( \gamma _{R}>\gamma
_{th}\right)  \label{IP2} \\
& +\Pr \left( \left. C_{\sec }\leq 0\right\vert \gamma _{R}<\gamma
_{th}\right) \Pr \left( \gamma _{R}<\gamma _{th}\right)  \notag \\
& \overset{(b)}{=}\Pr \left( \left. C_{\sec }\leq 0\right\vert \gamma
_{R}>\gamma _{th}\right) \Pr \left( \gamma _{R}>\gamma _{th}\right)  \notag
\\
& +\Pr \left( \gamma _{R}<\gamma _{th}\right) ,  \notag
\end{align}
where step (a) is attained by using both total probability and Bayes' rules,
while step (b) holds by {noting} that when $\gamma
_{R}<\gamma _{th} $, the satellite $\mathbf{R}$ fails to decode the received
message from $\mathbf{S}$ and {therefore} the communication
could not be established between the source and the destination (i.e. $%
\gamma _{D}=\gamma_{E_{2}}=0$). 
{It follows that $C_{\sec
	}=0$} {and consequently,} $\Pr \left(\left. C_{\sec }\leq
0\right\vert \gamma _{R}<\gamma _{th}\right)=1$.

On the other hand, by considering that 
\begin{align}
\Pr \left( \left. C_{\sec }\leq 0\right\vert \gamma _{R}>\gamma _{th}\right)
& =\Pr \left( \gamma _{R}>\gamma _{th}\right)  \label{Prob} \\
& -\Pr \left( \left. C_{\sec }>0\right\vert \gamma _{R}>\gamma _{th}\right) ,
\notag
\end{align}
{Substituting (\ref{Cs1}) and (\ref{Prob}) into (\ref{IP2}), the IP can be
	rewritten as } 
\begin{eqnarray}
\text{IP} &=&1-\Pr \left( C_{\sec }>0,\gamma _{R}>\gamma _{th}\right) ,
\label{IP3} \\
&=&1-\underset{\mathcal{I}}{\underbrace{\Pr \left( \gamma _{R}>\gamma
		_{E_{1}},\gamma _{R}>\gamma _{E_{2}},\gamma _{D}>\gamma _{E_{2}},\gamma
		_{R}>\gamma _{th}\right) }}.  \notag
\end{eqnarray}
Obviously, the term $\mathcal{I}$ can be expressed as the sum of six
different probabilities {(i.e. $\mathcal{I=}\sum_{i=1}^{6}\mathcal{I}_{i}$
	with $\mathcal{I}_{i}=\Pr \left( E_{i},{\gamma _{D}>\gamma_{E_{2}}}\right) $}
where the events $E_{i}$ are summarized in Table \ref{events}. 
\begin{table}[tbp]
	\caption{Events for $\mathcal{I}$.}
	\begin{center}
		\begin{tabular}{c|c||c|c}
			\hline
			&\textit{Event} &  & \textit{Event} 	\tabularnewline
			\hline
			$ E_{1} $ & $\gamma _{R}>\gamma _{E_{1}}>\gamma _{E_{2}}>\gamma _{th}$ & $ E_{4} $ & $\gamma _{R}>\gamma _{th}>\gamma _{E_{1}}>\gamma _{E_{2}}$    \tabularnewline
			\hline
			$ E_{2} $ & $\gamma _{R}>\gamma _{E_{2}}>\gamma _{E_{1}}>\gamma _{th}$ & $ E_{5} $ & $\gamma _{R}>\gamma _{E_{2}}>\gamma _{th}>\gamma _{E_{1}}$     \tabularnewline
			\hline
			$ E_{3} $ & $\gamma _{R}>\gamma _{th}>\gamma _{E_{2}}>\gamma _{E_{1}}$ & $ E_{6} $ & $\gamma _{R}>\gamma _{E_{1}}>\gamma _{th}>\gamma _{E_{2}}$    \tabularnewline
			\hline
		\end{tabular}
	\end{center}
	\label{events}
\end{table}
\begin{figure*}[!htb]
	\setcounter{mytempeqncnt}{\value{equation}} \setcounter{equation}{47}
	\par
	\begin{equation}
	\mathcal{I}_{1} =\int_{0}^{\infty }f_{g_{SP}}\left( u\right) du\int_{\gamma _{th}}^{\infty
	}F_{\gamma _{E_{1}}|g_{SP}=u}\left( y\right) \left[ 
	\begin{array}{c}
	f_{\gamma _{R}|g_{SP}=u}\left( y\right) \int_{\gamma _{th}}^{y}f_{\gamma
		_{E_{2}}}\left( z\right) F_{\gamma _{D}}^{c}\left( z\right) dz \\ 
	-F_{\gamma _{R}|g_{SP}=u}^{c}\left( y\right) f_{\gamma _{E_{2}}}\left(
	y\right) F_{\gamma _{D}}^{c}\left( y\right)%
	\end{array}%
	\right] dy.  \label{I1}
	\end{equation}%
\end{figure*}
%%%%%%%%%%%%%%%%%%%%%%%%%%%%%%%%%%%
\begin{figure*}[!htb]
	\setcounter{mytempeqncnt}{\value{equation}} \setcounter{equation}{48}
	\par
	\begin{equation}
	\mathcal{I}_{2} =\int_{0}^{\infty }f_{g_{SP}}\left( u\right) du\int_{\gamma _{th}}^{\infty
	}F_{\gamma _{R}|g_{SP}=u}^{c}\left( z\right) f_{\gamma _{E_{2}}}\left(
	z\right) F_{\gamma _{D}}^{c}\left( z\right) \left[ F_{\gamma
		_{E_{1}}|g_{SP}=u}\left( z\right) -F_{\gamma _{E_{1}}|g_{SP}=u}\left( \gamma
	_{th}\right) \right] dz.  \label{I2}
	\end{equation}
\end{figure*}
%%%%%%%%%%%%%
\begin{figure*}[!htb]
	\setcounter{mytempeqncnt}{\value{equation}} \setcounter{equation}{49}
	\par
	\begin{equation}
	\mathcal{I}_{3} =\int_{0}^{\infty }F_{\gamma _{R}|g_{SP}=u}^{c}\left( \gamma _{th}\right)
	f_{g_{SP}}\left( u\right) \int_{0}^{\gamma _{th}}F_{\gamma
		_{E_{1}}|g_{SP}=u}\left( y\right) f_{\gamma _{E_{2}}}\left( y\right)
	F_{\gamma _{D}}^{c}\left( y\right) dudy.  \label{I3}
	\end{equation}%
\end{figure*}
%%%%%%%%%%%%%%%%%%%%%%
\begin{figure*}[!htb]
	\setcounter{mytempeqncnt}{\value{equation}} \setcounter{equation}{50}
	\par
	\begin{equation}
	\mathcal{I}_{4} =\int_{0}^{\infty }F_{\gamma _{R}|g_{SP}=u}^{c}\left( \gamma _{th}\right)
	f_{g_{SP}}\left( u\right) \int_{0}^{\gamma _{th}}f_{\gamma _{E_{2}}}\left(
	z\right) F_{\gamma _{D}}^{c}\left( z\right) \left[ F_{\gamma
		_{E_{1}}|g_{SP}=u}\left( \gamma _{th}\right) -F_{\gamma
		_{E_{1}}|g_{SP}=u}\left( z\right) \right] dudz.   \label{I4}
	\end{equation}
\end{figure*}
%%%%%%%%%%%%%%%%%%%%%%%%%%%
\begin{figure*}[!htb]
	\setcounter{mytempeqncnt}{\value{equation}} \setcounter{equation}{51}
	\par
	\begin{equation}
	\mathcal{I}_{5} =\int_{0}^{\infty }F_{\gamma _{E_{1}}|g_{SP}=u}\left( \gamma _{th}\right)
	f_{g_{SP}}\left( u\right) \left[ \int_{\gamma _{th}}^{\infty }f_{\gamma
		_{E_{2}}}\left( x\right) F_{\gamma _{D}}^{c}\left( x\right) -\int_{\gamma
		_{th}}^{\infty }F_{\gamma _{R}|g_{SP}=u}\left( x\right) f_{\gamma
		_{E_{2}}}\left( x\right) F_{\gamma _{D}}^{c}\left( x\right) \right] dudx. \label{I5}
	\end{equation}%
\end{figure*}
%%%%%%%%%%%%%%%%%%%%%%%%%%%%%
\begin{figure*}[!htb]
	\setcounter{mytempeqncnt}{\value{equation}} \setcounter{equation}{52}
	\par
	\begin{equation}
	\mathcal{I}_{6} =\int_{0}^{\infty }f_{g_{SP}}\left( u\right) \int_{\gamma _{th}}^{\infty
	}f_{\gamma _{R}|g_{SP}=u}\left( x\right) \left[ F_{\gamma
		_{E_{1}}|g_{SP}=u}\left( x\right) -F_{\gamma _{E_{1}}|g_{SP}=u}\left( \gamma
	_{th}\right) \right] dxdu\times \int_{0}^{\gamma _{th}}f_{\gamma
		_{E_{2}}}\left( z\right) F_{\gamma _{D}}^{c}\left( z\right) dz.  \label{I6}
	\end{equation}%
\end{figure*}
Eqs. (\ref{I1}), (\ref{I3}), and (\ref{I5}) are obtained using integration
by parts alongside some algebraic manipulations, while (\ref{I2}), (\ref{I4}%
), and (\ref{I6}) can be achieved {by applying the basic definition of the
	CDF.}

By performing a summation of (\ref{I1})-(\ref{I6}) and substituting $%
\mathcal{I}$ into (\ref{IP3}), (\ref{IP_new_expression}) is attained, which
concludes the proof of {\textbf{Lemma \ref{lemma1}}}.

\section*{Appendix B: proof of theorem 1}

To prove the IP expression given in (\ref{IP1}) in both the absence and
presence of friendly jammer {cases}, it is mandatory to compute $\mathcal{J%
}_{1}\left( y,u\right) $, $\mathcal{J}_{2}\left( y\right) $, and $\mathcal{R}%
_{1}\left( u\right) $. As one can see from (\ref{J1}) and (\ref{J2}), to
compute $\mathcal{J}_{1}\left(y,u\right)$ it is sufficient to derive the
conditional CDFs of (\ref{GammaR}), (\ref{gammaE1_without_jammer}) and (\ref%
{gammaE1}), while $\mathcal{J}_{2}\left(y\right)$ can be attained using (\ref%
{pdf_hop1}) and (\ref{cdf_hop1}).

To do so, we start by computing the conditional CDF of $\gamma_{R}$ for a
given $g_{SP}$ as follows 
\begin{align}
F_{\gamma _{R}|g_{SP}=u}\left( y\right) & =\Pr \left( \gamma _{R}\leq
y\left\vert g_{SP}=u\right. \right)  \label{F_gammaRk_exp1} \\
& \overset{\left( a\right) }{=}F_{g_{SR}}\left( \frac{y}{\Omega \left(
	u\right) }\right) ,  \notag
\end{align}%
where $\Omega \left( u\right) =\overline{\gamma }_{S}$ if $u\leq \sigma _{S}$
and $\Omega \left( u\right) =\frac{\overline{\gamma }_{I}}{u}$ if $u>\sigma
_{S}$, with $\sigma _{S}$ is defined in \textbf{Theorem \ref{theorem1}}, and 
{Step $\left( a\right) $ holds by using (\ref{GammaR}).}

\begin{itemize}
	\item \textbf{Absence of friendly jammer {case}.}
\end{itemize}

Using {(}\ref{gammaE1_without_jammer}{)}, the conditional CDF of $\gamma
_{E_{1}}$ for a given $g_{SP}$ can be expressed as \hspace*{-0.35cm} 
\begin{align}
F_{\gamma _{E_{1}}|g_{SP}=u}\left( y\right) & =\Pr \left( \min \left( 
\overline{\gamma }_{S},\frac{\overline{\gamma }_{I}}{u}\right)
g_{SE_{1}}\leq y\right)  \label{F_gammaE1_nojammer} \\
& =F_{g_{SE_{1}}}\left( \frac{y}{\Omega \left( u\right) }\right) .  \notag
\end{align}

Substituting (\ref{F_gammaRk_exp1}) and (\ref{F_gammaE1_nojammer}) into (\ref%
{J1}), yields 
\begin{equation}
\mathcal{J}_{1}\left( y,u\right) =\frac{1}{\Omega \left( u\right) }%
f_{g_{SR}}\left( \frac{y}{\Omega \left( u\right) }\right)
F_{g_{SE_{1}}}\left( \frac{y}{\Omega \left( u\right) }\right) {.}
\end{equation}%
{Now,} substituting (\ref{pdf_hop1}) into (\ref{cdf_hop1}), the term $%
\mathcal{J}_{1}\left( y,u\right) $ can be rewritten as 
\begin{align}
\mathcal{J}_{1}\left( y,u\right) & =\frac{\Delta _{SE_{1}}\Delta _{SR}e^{-%
		\frac{\upsilon _{SR}}{\Omega \left( u\right) }y}}{\Omega \left( u\right) }%
\sum_{n_{2}=0}^{m_{SE_{1}}-1}\frac{\phi _{SE_{1}}^{(n_{2})}}{\upsilon
	_{SE_{1}}^{n_{2}+1}}{}  \label{J1_nojammer} \\
& \times \gamma _{inc}\left( {\small n}_{2}{\small +1},\frac{\upsilon
	_{SE_{1}}}{\Omega \left( u\right) }y\right) \sum_{n_{3}=0}^{m_{SR}-1}\frac{%
	\phi _{SR}^{(n_{3})}}{\Omega ^{n_{3}}\left( u\right) }y^{n_{3}}.  \notag
\end{align}%
Next, the term $\mathcal{J}_{2}\left( y\right) $ given in (\ref{J2}) can be
expressed as

\begin{equation}
\mathcal{J}_{2}\left( y\right) =F_{\gamma _{E_{2}}}\left( y\right)
-\int_{z=0}^{y}f_{\gamma _{E_{2}}}\left( z\right) F_{\gamma _{D}}\left(
z\right) dz.  \label{J2_exp2}
\end{equation}%
{By substituting (\ref{pdfx}) and (\ref{cdfx}) into (\ref{J2_exp2}), we get} 
\begin{align}
\mathcal{J}_{2}\left( y\right) & =F_{\gamma _{E_{2}}}\left( y\right) -\theta
\notag \\
& \times \int_{z=0}^{y}\frac{1}{z}G_{1,3}^{3,0}\left( \frac{\Upsilon
	_{E_{2}}z}{\mu _{1}^{(E_{2})}}\left\vert 
\begin{array}{c}
-;\xi _{E_{2}}^{2}+1 \\ 
\xi _{E_{2}}^{2},\alpha _{E_{2}},\beta _{E_{2}};-%
\end{array}%
\right. \right)  \notag \\
& \times G_{2,4}^{3,1}\left( \frac{\Upsilon _{D}z}{\mu _{1}^{(D)}}\left\vert 
\begin{array}{c}
1;\xi _{D}^{2}+1 \\ 
\xi _{D}^{2},\alpha _{D},\beta _{D};0%
\end{array}%
\right. \right) dz,  \label{J2_exp3}
\end{align}%
where $\theta $ is defined in {\textbf{Theorem \ref{theorem1}}}.

To compute (\ref{J2_exp3}), we can express one of the two Meijer's
G-functions as an infinite sum of the respective integrand's residues
evaluated at the appropriate poles \cite[Theorem 1.5]{kilbas}. That is, 
\begin{align}
G_{2,4}^{3,1}\left( y\left\vert 
\begin{array}{c}
1-\tau ;\xi _{Z}^{2}+1 \\ 
\xi _{Z}^{2},\alpha _{Z},\beta _{Z};-\tau%
\end{array}%
\right. \right) & =e^{(Z,\tau )}y^{\xi _{Z}^{2}}  \label{residd} \\
& +\sum_{k=0}^{\infty }e^{(Z,\tau ,k)}\left( \beta _{Z},\alpha _{Z}\right)
y^{\alpha _{Z}+k}  \notag \\
& +\sum_{k=0}^{\infty }e^{(Z,\tau ,k)}\left( \alpha _{Z},\beta _{Z}\right)
y^{\beta _{Z}+k},  \notag
\end{align}%
where $\tau =\{0,q\},$ $q$ {takes a value in the set} $\{\xi _{D}^{2},\alpha
_{D}+k,\beta _{D}+k\},$ $e^{(Z,\tau )},$ and $e^{(Z,\tau ,k)}\left(
.,.\right) $ are defined in \textbf{Theorem \ref{theorem1}}.

Substituting $\left( \ref{residd}\right) $ into (\ref{J2_exp3}), {yields} 
\begin{align}
\mathcal{J}_{2}\left( y\right) & =F_{\gamma _{E_{2}}^{(k)}}\left( y\right)
-\theta  \label{J2_final} \\
& \times \left[ 
\begin{array}{c}
e^{(D,0)}\mathcal{P}\left( \xi _{D}^{2},y\right) \\ 
+\sum_{k=0}^{\infty }e^{(D,0,k)}\left( \beta _{D},\alpha _{D}\right) 
\mathcal{P}\left( \alpha _{D}+k,y\right) \\ 
+\sum_{k=0}^{\infty }e^{(D,0,k)}\left( \alpha _{D},\beta _{D}\right) 
\mathcal{P}\left( \beta _{D}+k,y\right)%
\end{array}%
\right] ,  \notag
\end{align}%
with 
\begin{align}
\mathcal{P}\left( q,y\right) & =\left( \frac{\Upsilon _{D}}{\mu _{1}^{(D)}}%
\right) ^{q}\int_{0}^{y}z^{q-1}  \label{Pq} \\
& \times G_{1,3}^{3,0}\left( \frac{\Upsilon _{E_{2}}z}{\mu _{1}^{(E_{2})}}%
\left\vert 
\begin{array}{c}
-;\xi _{E_{2}}^{2}+1 \\ 
\xi _{E_{2}}^{2},\alpha _{E_{2}},\beta _{E_{2}};-%
\end{array}%
\right. \right) dz  \notag \\
& \overset{\left( a\right) }{=}\left( \frac{\Upsilon _{D}y}{\mu _{1}^{(D)}}%
\right) ^{q}G_{2,4}^{3,1}\left( \frac{\Upsilon _{E_{2}}y}{\mu _{1}^{(E_{2})}}%
\left\vert 
\begin{array}{c}
1-q;\xi _{E_{2}}^{2}+1 \\ 
\xi _{E_{2}}^{2},\alpha _{E_{2}},\beta _{E_{2}};-q%
\end{array}%
\right. \right) ,  \notag
\end{align}%
where step $\left( a\right) $ follows using \cite[Eq. 07.34.21.0003.01]%
{Wolfram}.

Replacing, (\ref{J1_nojammer}) and (\ref{J2_final}) into (\ref{U1}), one can
obtain

\begin{align}
\mathcal{R}_{1}\left( u\right) & =\frac{\xi _{_{E_{2}}}^{2}\mathcal{L}%
	_{1}^{\left( 0\right) }\left( u\right) }{\Gamma \left( \alpha
	_{E_{2}}\right) \Gamma \left( \beta _{E_{2}}\right) }
\label{U1_nojammer_exp1} \\
& -\theta \left[ 
\begin{array}{c}
e^{(D,0)}\mathcal{L}_{1}^{\left( \xi _{D}^{2}\right) }\left( u\right) \\ 
+\sum_{k=0}^{\infty }e^{(D,0,k)}\left( \beta _{D},\alpha _{D}\right) 
\mathcal{L}_{1}^{\left( \alpha _{D}+k\right) }\left( u\right) \\ 
+\sum_{k=0}^{\infty }e^{(D,0,k)}\left( \alpha _{D},\beta _{D}\right) 
\mathcal{L}_{1}^{\left( \beta _{D}+k\right) }\left( u\right)%
\end{array}%
\right] ,  \notag
\end{align}%
where 
\begin{align}
\mathcal{L}_{1}^{\left( \tau \right) }\left( u\right) & =\Delta
_{SE_{1}}\Delta _{SR}\left( \frac{\Upsilon _{D}}{\mu _{1}^{(D)}}\right)
^{\tau }  \label{L1_exp1} \\
& \times \sum_{n_{2}=0}^{m_{SE_{1}}-1}\frac{\phi _{SE_{1}}^{(n_{2})}}{%
	\upsilon _{SE_{1}}^{n_{2}+1}}\sum_{n_{3}=0}^{m_{SR}-1}\frac{\phi
	_{SR}^{(n_{3})}}{\Omega ^{n_{3}+1}\left( u\right) }{}  \notag \\
& \times \int_{y=\gamma _{th}}^{\infty }y^{\tau +n_{3}}\gamma _{inc}\left( 
{\small n}_{2}{\small +1},\frac{\upsilon _{SE_{1}}}{\Omega \left( u\right) }%
y\right) e^{-\frac{\upsilon _{SR}}{\Omega \left( u\right) }y}  \notag \\
& \times G_{2,4}^{3,1}\left( \frac{\Upsilon _{E_{2}}y}{\mu _{1}^{(E_{2})}}%
\left\vert 
\begin{array}{c}
1-\tau ;\xi _{E_{2}}^{2}+1 \\ 
\xi _{E_{2}}^{2},\alpha _{E_{2}},\beta _{E_{2}};-\tau%
\end{array}%
\right. \right) dy.  \notag
\end{align}%
{Now,} replacing (\ref{residd}) into (\ref{L1_exp1}), yields 
\begin{align}
\mathcal{L}_{1}^{\left( \tau \right) }\left( u\right) & =\Delta
_{SE_{1}}\Delta _{SR}\left( \frac{\Upsilon _{D}}{\mu _{1}^{(D)}}\right)
^{\tau }  \label{L1_exp2} \\
& \times \sum_{n_{2}=0}^{m_{SE_{1}}-1}\frac{\phi _{SE_{1}}^{(n_{2})}}{%
	\upsilon _{SE_{1}}^{n_{2}+1}}\sum_{n_{3}=0}^{m_{SR}-1}\frac{\phi
	_{SR}^{(n_{3})}}{\Omega ^{n_{3}+1}\left( u\right) }{}  \notag \\
& \times \left[ 
\begin{array}{c}
e^{(E_{2},\tau )}\digamma _{1}^{\left( n_{2},n_{3},\xi _{E_{2}}^{2},\tau
	\right) }\left( u\right) \\ 
+\sum_{k=0}^{\infty }e^{(E_{2},\tau ,k)}\left( \beta _{E_{2}},\alpha
_{E_{2}}\right) \digamma _{1}^{\left( n_{2},n_{3},\alpha _{E_{2}}+k,\tau
	\right) }\left( u\right) \\ 
+\sum_{k=0}^{\infty }e^{(E_{2},\tau ,k)}\left( \alpha _{E_{2}},\beta
_{E_{2}}\right) \digamma _{1}^{\left( n_{2},n_{3},\beta _{E_{2}}+k,\tau
	\right) }\left( u\right)%
\end{array}%
\right] ,  \notag
\end{align}%
with 
\begin{align}
\digamma _{1}^{\left( n_{2},n_{3},a,\tau \right) }\left( u\right) & =\left( 
\frac{\Upsilon _{E_{2}}}{\mu _{1}^{(E_{2})}}\right) ^{a}\int_{y=\gamma
	_{th}}^{\infty }y^{\tau +n_{3}+a}  \label{F1_exp1} \\
& \times e^{-\frac{\upsilon _{SR}}{\Omega \left( u\right) }y}\gamma
_{inc}\left( {\small n}_{2}{\small +1},\frac{\upsilon _{SE_{1}}}{\Omega
	\left( u\right) }y\right) dy,  \notag
\end{align}%
where $a$ {belongs to the set} $\{\xi _{E_{2}}^{2},\alpha _{E_{2}}+k,\beta
_{E_{2}}+k\}.$

Using \cite[Eq. (06.06.26.0004.01)]{Wolfram}, (\ref{F1_exp1}) can be
expressed as 
\begin{align}
\digamma _{1}^{\left( n_{2},n_{3},a,\tau \right) }\left( u\right) & =\left( 
\frac{\Upsilon _{E_{2}}}{\mu _{1}^{(E_{2})}}\right) ^{a}\left( \frac{\Omega
	\left( u\right) }{\upsilon _{SR}}\right) ^{l+1}  \label{F1_exp2} \\
& \times G_{2,2}^{1,2}\left( \frac{\upsilon _{SE_{1}}}{\upsilon _{SR}}%
\left\vert 
\begin{array}{c}
\left( 1,0\right) ,\left( -l,\varsigma \right) ;- \\ 
\left( {\small n}_{2}{\small +1,0}\right) ;\left( 0,0\right)%
\end{array}%
\right. \right) ,  \notag
\end{align}%
with $l=n_{3}+\tau +a$ and $\varsigma =\frac{\gamma _{th}\upsilon _{SR}}{%
	\Omega \left( u\right) }.$

Now, substituting $\left( \ref{U1_nojammer_exp1}\right) $ into $\left( \ref%
{IP_new_expression}\right) ,$ the IP in the absence of a friendly jammer can
be written as 
\begin{align}
\text{IP}_{A}& =1-\frac{\xi _{E_{2}}^{2}\mathcal{D}_{A}^{\left( 0\right) }}{%
	\Gamma \left( \alpha _{E_{2}}\right) \Gamma \left( \beta _{E_{2}}\right) }
\label{IP_without_jammer_exp1} \\
& +\theta \left[ 
\begin{array}{c}
e^{(D,0)}\mathcal{D}_{A}^{\left( \xi _{D}^{2}\right) } \\ 
+\sum_{k=0}^{\infty }e^{(D,0,k)}\left( \beta _{D},\alpha _{D}\right) 
\mathcal{D}_{A}^{\left( \alpha _{D}+k\right) } \\ 
+\sum_{k=0}^{\infty }e^{(D,0,k)}\left( \alpha _{D},\beta _{D}\right) 
\mathcal{D}_{A}^{\left( \beta _{D}+k\right) }%
\end{array}%
\right] ,  \notag
\end{align}%
where 
\begin{equation}
\mathcal{D}_{A}^{\left( \tau \right) }=\frac{\Upsilon _{D}^{\tau }}{\left(
	\mu _{1}^{(D)}\right) ^{\tau }}\int_{u=0}^{\infty }f_{g_{SP}}\left( u\right) 
\mathcal{L}_{1}^{\left( \tau \right) }\left( u\right) du,
\label{Dtau_without_jammer_exp1}
\end{equation}%
To compute $\left( \ref{Dtau_without_jammer_exp1}\right) ,$ we need to
replace $\Omega \left( u\right) $ by its values. Therefore, when $u\leq
\sigma _{S}$ the term $\mathcal{L}_{1}^{\left( \tau \right) }\left( u\right) 
$ becomes constant. Thus, 
\begin{equation}
\mathcal{D}_{A}^{\left( \tau \right) }=\frac{\Upsilon _{D}^{\tau }}{\left(
	\mu _{1}^{(D)}\right) ^{\tau }}\left[ \mathcal{L}_{2}^{\left( \tau \right)
}\left( \overline{\gamma }_{S}\right) F_{g_{SP}}\left( \sigma _{S}\right)
+\lambda _{SP}\mathcal{L}_{2}^{\left( \tau \right) }\left( \overline{\gamma }%
_{I}\right) \right] ,  \label{Dtau_without_jammer_exp2}
\end{equation}%
with $\mathcal{L}_{2}^{\left( \tau \right) }\left( \overline{\gamma }%
_{I}\right) =\int_{\sigma _{S}}^{\infty }e^{-\lambda _{SP}u}\mathcal{L}%
_{1}^{\left( \tau \right) }\left( u\right) du.$

The term $\mathcal{L}_{2}^{\left( \tau \right) }\left( \overline{\gamma }%
_{S}\right) $ and $\mathcal{L}_{2}^{\left( \tau \right) }\left( \overline{%
	\gamma }_{I}\right) $ can be obtained by replacing $\Omega \left( u\right) =%
\overline{\gamma }_{S}$ and $\Omega \left( u\right) =\frac{\overline{\gamma }%
	_{I}}{u}$ in (\ref{L1_exp2}) and (\ref{F1_exp2}), respectively. Therefore, \ 
$\mathcal{L}_{2}^{\left( \tau \right) }\left( \overline{\gamma }_{I}\right) $
can be rewritten as 
\begin{align}
\mathcal{L}_{2}^{\left( \tau \right) }\left( \overline{\gamma }_{I}\right) &
=\Delta _{SE_{1}}\Delta _{SR}\left( \frac{\Upsilon _{D}}{\mu _{1}^{(D)}}%
\right) ^{\tau } \\
& \times \sum_{n_{3}=0}^{m_{SR}-1}\frac{\phi _{SR}^{(n_{3})}}{\overline{%
		\gamma }_{I}^{n_{3}+1}}{}\sum_{n_{2}=0}^{m_{SE_{1}}-1}\frac{\phi
	_{SE_{1}}^{(n_{2})}}{\upsilon _{SE_{1}}^{n_{2}+1}}  \notag \\
& \times \left[ 
\begin{array}{c}
{\small e}^{(E_{2},\tau )}\digamma _{2}^{\left( n_{2},n_{3},\xi
	_{E_{2}}^{2},\tau \right) } \\ 
+\sum_{k=0}^{\infty }{\small e}^{(E_{2},\tau ,k)}\left( {\small \beta }%
_{E_{2}}{\small ,\alpha }_{E_{2}}\right) {\small \digamma }_{2}^{\left( 
	{\small n}_{2}{\small ,n}_{3}{\small ,\alpha }_{E_{2}}{\small +k,\tau }%
	\right) } \\ 
+\sum_{k=0}^{\infty }{\small e}^{(E_{2},\tau ,k)}\left( {\small \alpha }%
_{E_{2}}{\small ,\beta }_{E_{2}}\right) {\small \digamma }_{2}^{\left( 
	{\small n}_{2}{\small ,n}_{3}{\small ,\beta }_{E_{2}}{\small +k,\tau }%
	\right) }%
\end{array}%
\right] ,  \notag
\end{align}%
where 
\begin{align}
\digamma _{2}^{\left( n_{2},n_{3},a,\tau \right) }& =\left( \frac{\Upsilon
	_{E_{2}}}{\mu _{1}^{(E_{2})}}\right) ^{a}\int_{\sigma _{S}}^{\infty
}e^{-\lambda _{SP}u}  \label{F2_exp2} \\
& \times u^{n_{3}+1}\digamma _{1}^{\left( n_{2},n_{3},\alpha _{E_{2}}+k,\tau
	\right) }\left( u\right) du  \notag \\
& =\left( \frac{\overline{\gamma }_{I}}{\upsilon _{SR}}\right) ^{\tau
	+n_{3}+a+1}\frac{1}{2\pi j}\int_{\mathcal{C}_{s}}\frac{\Gamma \left( {\small %
		n}_{2}{\small +1+s}\right) }{\Gamma \left( 1-{\small s}\right) }  \notag \\
& \times \Gamma \left( -{\small s}\right) \left( \frac{\upsilon _{SE_{1}}}{%
	\upsilon _{SR}}\right) ^{-s}\int_{\sigma _{S}}^{\infty }u^{-\tau
	-a}e^{-\lambda _{SP}u}  \notag \\
& \times \Gamma \left( 1+n_{3}+\tau +a-s,\epsilon _{I}\upsilon _{SR}u\right)
dsdu,  \notag
\end{align}%
where $\epsilon _{I}$ is defined in \textbf{Theorem \ref{theorem1}}.

Using \cite[Eq. (06.06.26.0005.01)]{Wolfram}, one obtains).%
\begin{align}
\digamma _{2}^{\left( n_{2},n_{3},a,\tau \right) }& =\left( \frac{\Upsilon
	_{E_{2}}}{\mu _{1}^{(E_{2})}}\right) ^{a}\left( \frac{\overline{\gamma }_{I}%
}{\upsilon _{SR}}\right) ^{\tau +n_{3}+a+1}\frac{\lambda _{SP}^{a+\tau -1}}{%
	\left( 2\pi j\right) ^{2}}  \label{F2_final} \\
& \times \int_{\mathcal{C}_{s}}\frac{\Gamma \left( {\small n}_{2}{\small +1+s%
	}\right) \Gamma \left( -{\small s}\right) }{\Gamma \left( 1-{\small s}%
	\right) }\left( \frac{\upsilon _{SE_{1}}}{\upsilon _{SR}}\right) ^{-s} 
\notag \\
& \times \int_{\mathcal{C}_{w}}\frac{\Gamma \left( -\tau -a-w+1,\lambda
	_{SP}\sigma _{S}\right) \left( \epsilon _{I}\upsilon _{SR}\right) ^{-w}}{%
	w\lambda _{SP}^{-w}}  \notag \\
& \times \Gamma \left( \tau +n_{3}+a+1-s+w\right) dwds,  \notag
\end{align}

\begin{itemize}
	\item \textbf{Presence of a friendly jammer}.
\end{itemize}

Using (\ref{gammaE1}), the conditional CDF of $\gamma _{E_{1}}^{(J)}$ for a
given $g_{SP}$ can be expressed as 
\begin{subequations}
	\begin{align}
	F_{\gamma _{E_{1}}^{(J)}|g_{SP}=u}\left( y\right) & =\Pr \left( \mathcal{U}%
	_{S}\leq y\left( \mathcal{U}_{S_{J}}+1\right) \right)  \label{F_gammaE1_exp1}
	\\
	& =1-\frac{y}{\Omega \left( u\right) }\int_{0}^{\infty }f_{g_{SE_{1}}}\left( 
	\frac{y\left( t+1\right) }{\Omega \left( u\right) }\right)
	\label{F_gammaE1_exp1_Step2} \\
	& \times F_{\mathcal{U}_{S_{J}}}\left( t\right) dt,  \notag
	\end{align}%
	where (\ref{F_gammaE1_exp1_Step2}) holds by using integration by parts.
	
	On the other hand, the CDF of $\mathcal{U}_{S_{J}}$ is given by 
\end{subequations}
\begin{align}
F_{\mathcal{U}_{S_{J}}}\left( t\right) & =\Pr \left( \min \left( \overline{
	\gamma }_{S_{J}},\frac{\overline{\gamma }_{I}}{g_{S_{J}P}}\right)
g_{S_{J}E_{1}}\leq t\right)  \label{F_UJ_exp1} \\
& =F_{g_{S_{J}E_{1}}}\left( \frac{t}{\overline{\gamma }_{S_{J}}}\right)
F_{g_{S_{J}P}}\left( \sigma _{S_{J}}\right)  \notag \\
& +\underset{\mathcal{K}_{1}}{\underbrace{\int_{\sigma _{S_{J}}}^{\infty
		}F_{g_{S_{J}E_{1}}}\left( \frac{t}{\overline{\gamma }_{I}}y\right)
		f_{g_{S_{J}P}}\left( y\right) dy}}.  \notag
\end{align}%
where $\sigma _{S_{J}}$ is defined in \textbf{Theorem \ref{theorem1}}.

Using (\ref{cdf_hop1}) and \cite[Eq. (06.06.26.0004.01)]{Wolfram}, the term $%
\mathcal{K}_{1}$ can be expressed as%
\begin{align}
\mathcal{K}_{1}& =\lambda _{S_{J}P}\Delta
_{S_{J}E_{1}}{}\sum_{n_{1}=0}^{m_{S_{J}E_{1}}-1}\frac{\phi
	_{S_{J}E_{1}}^{(n_{1})}}{\upsilon _{S_{J}E_{1}}^{n_{1}+1}}  \label{K1_final}
\\
& \times \int_{\sigma _{S_{J}}}^{\infty }e^{-\lambda
	_{S_{J}P}y}G_{1,2}^{1,1}\left( \frac{\upsilon _{S_{J}E_{1}}t}{\overline{
		\gamma }_{I}}y\left\vert 
\begin{array}{c}
1;- \\ 
n_{1}+1;0%
\end{array}%
\right. \right) dy  \notag \\
& =\Delta _{S_{J}E_{1}}{}\sum_{n_{1}=0}^{m_{S_{J}E_{1}}-1}\frac{\phi
	_{S_{J}E_{1}}^{(n_{1})}}{\upsilon _{S_{J}E_{1}}^{n_{1}+1}}\mathcal{G}%
_{1}\left( t\right) ,  \notag
\end{align}%
where $\mathcal{G}_{1}\left( t\right) =G_{2,2}^{1,2}\left( \frac{\upsilon
	_{S_{J}E_{1}}t}{\overline{\gamma }_{I}\lambda _{S_{J}P}}\left\vert 
\begin{array}{c}
\left( 1,0\right) ,\left( 0,\sigma _{S_{J}}\lambda _{S_{J}P}\right) ;- \\ 
\left( n_{1}+1,0\right) ;\left( 0,0\right)%
\end{array}%
\right. \right) .$

Substituting (\ref{K1_final}) into (\ref{F_UJ_exp1}), {yields the CDF\ of $%
	\mathcal{U}_{S_{J}}$} 
\begin{align}
F_{\mathcal{U}_{S_{J}}}\left( t\right) & =F_{g_{S_{J}E_{1}}}\left( \frac{t}{%
	\overline{\gamma }_{S_{J}}}\right) F_{g_{S_{J}P}}\left( \sigma
_{S_{J}}\right)   \label{F_UJ_final} \\
& +\Delta _{S_{J}E_{1}}{}\sum_{n_{1}=0}^{m_{S_{J}E_{1}}-1}\frac{\phi
	_{S_{J}E_{1}}^{(n_{1})}}{\upsilon _{S_{J}E_{1}}^{n_{1}+1}}\mathcal{G}%
_{1}\left( t\right) .  \notag
\end{align}%
Now, replacing (\ref{F_UJ_final}) into (\ref{F_UJ_exp1}), {the CDF of $%
	\gamma _{E_{1}}$ can be expressed as} 
\begin{align}
F_{\gamma _{E_{1}}^{(J)}|g_{SP}=u}\left( y\right) & =1-\frac{y\Delta
	_{SE_{1}}}{\Omega \left( u\right) }e^{-\frac{\upsilon _{SE_{1}}}{\Omega
		\left( u\right) }y}\sum_{n_{2}=0}^{m_{SE_{1}}-1}\frac{\phi
	_{SE_{1}}^{(n_{2})}y^{n_{2}}}{\Omega ^{n_{2}}\left( u\right) }
\label{F_gammaE1_exp2} \\
& \times \sum_{p=0}^{n_{2}}\binom{n_{2}}{p}\Delta
_{S_{J}E_{1}}{}\sum_{n_{1}=0}^{m_{S_{J}E_{1}}-1}\frac{\phi
	_{S_{J}E_{1}}^{(n_{1})}}{\upsilon _{S_{J}E_{1}}^{n_{1}+1}}  \notag \\
& \times \left[ F_{g_{S_{J}P}}\left( \sigma _{S_{J}}\right) \mathcal{V}%
_{1}\left( y,u\right) +\mathcal{V}_{2}\left( y,u\right) \right] dt,  \notag
\end{align}%
where 
\begin{equation}
\mathcal{V}_{1}\left( y,u\right) =\int_{0}^{\infty }t^{p}e^{-\frac{\upsilon
		_{SE_{1}}y}{\Omega \left( u\right) }t}\gamma _{inc}\left( n_{1}+1,\frac{%
	\upsilon _{S_{J}E_{1}}}{\overline{\gamma }_{S_{J}}}t\right) dt,
\end{equation}%
and 
\begin{equation}
\mathcal{V}_{2}\left( y,u\right) =\int_{0}^{\infty }t^{p}e^{-\frac{\upsilon
		_{SE_{1}}y}{\Omega \left( u\right) }t}\mathcal{G}_{1}\left( t\right) dt,
\end{equation}%
Using \cite[Eq.(07.34.21.0088.01)]{Wolfram}, the term $\mathcal{V}_{1}\left(
y,u\right) $ can be expressed as 
\begin{equation}
\mathcal{V}_{1}\left( y,u\right) =\left( \frac{\Omega \left( u\right) }{%
	\upsilon _{SE_{1}}y}\right) ^{p+1}G_{2,2}^{1,2}\left( \frac{\upsilon
	_{S_{J}E_{1}}\Omega \left( u\right) }{\upsilon _{SE_{1}}\overline{\gamma }%
	_{S_{J}}y}\left\vert 
\begin{array}{c}
-p,1;- \\ 
n_{1}+1;0%
\end{array}%
\right. \right) ,  \label{V1}
\end{equation}%
{while} the term $\mathcal{V}_{2}\left( y,u\right) $ can be evaluated as 
\begin{equation}
\mathcal{V}_{2}\left( y,u\right) =\left( \frac{\Omega \left( u\right) }{%
	\upsilon _{SE_{1}}y}\right) ^{p+1}\mathcal{G}_{2}\left( y,u\right) ,
\label{V2}
\end{equation}%
where $\mathcal{G}_{2}\left( y,u\right) =G_{3,2}^{1,3}\left( \rho \left\vert 
\begin{array}{c}
\left( 1,0\right) ,\left( 0,l\right) ,\left( -p,0\right) ;- \\ 
\left( n_{1}+1,0\right) ;\left( 0,0\right) 
\end{array}%
\right. \right) ,$ $l=\sigma _{S_{J}}\lambda _{S_{J}P},$ and $\rho =\frac{%
	\upsilon _{S_{J}E_{1}}\Omega \left( u\right) }{\lambda _{S_{J}P}\upsilon
	_{SE_{1}}\overline{\gamma }_{I}y}.$

Replacing (\ref{V1}) and (\ref{V2}) into (\ref{F_gammaE1_exp2}), and then
substituting the obtained expression of $F_{\gamma _{E_{1}}|g_{SP}=u}\left(
y\right) $ alongside (\ref{F_gammaRk_exp1}) {into (\ref{J1}), the term $%
	\mathcal{J}_{1}\left( y,u\right) $ can be expressed as}

\begin{align}
\mathcal{J}_{1}\left( y,u\right) & =\Delta _{SR}\sum_{n_{3}=0}^{m_{SR}-1}%
\frac{\phi _{SR}^{(n_{3})}y^{n_{3}}}{\Omega ^{n3+1}\left( u\right) }e^{-%
	\frac{\upsilon _{SR}}{\Omega \left( u\right) }y}  \label{J1_final} \\
& -\Delta _{SR}\Delta _{S_{J}E_{1}}\Delta
_{SE_{1}}\sum_{n_{1}=0}^{m_{S_{J}E_{1}}-1}\sum_{n_{2}=0}^{m_{SE_{1}}-1}%
\sum_{p=0}^{n_{2}}\sum_{n_{3}=0}^{m_{SR}-1}  \notag \\
& \times \frac{\mathcal{B}^{(n_{1},n_{2},n_{3},p)}}{\Omega
	^{n_{2}+n_{3}-p+1}\left( u\right) }e^{-\frac{\zeta }{\Omega \left( u\right) }%
	y}y^{n_{2}+n_{3}-p}  \notag \\
& \times \left[ 
\begin{array}{c}
F_{g_{S_{J}P}}\left( \sigma _{S_{J}}\right) {}G_{2,2}^{1,2}\left( \frac{%
	\upsilon _{S_{J}E_{1}}\Omega \left( u\right) }{\upsilon _{SE_{1}}\overline{%
		\gamma }_{S_{J}}y}\left\vert 
\begin{array}{c}
-p,1;- \\ 
n_{1}+1;0%
\end{array}%
\right. \right) \\ 
+\mathcal{G}_{2}\left( y,u\right)%
\end{array}%
\right] ,  \notag
\end{align}%
where $\mathcal{B}^{(n_{1},n_{2},n_{3},p)}$ is defined in \textbf{Theorem \ref{theorem1}}.

Now, substituting (\ref{J1_final}) and $\left( \ref{J2_final}\right) $ into (%
\ref{U1}), the term $\mathcal{U}_{1}\left( u\right) $ can be expressed as 
\begin{align}
\mathcal{U}_{1}\left( u\right) & =\frac{\xi _{E_{2}}^{2}\mathcal{T}%
	_{1}^{\left( 0\right) }\left( u\right) }{\Gamma \left( \alpha
	_{E_{2}}\right) \Gamma \left( \beta _{E_{2}}\right) }  \label{I1_exp2} \\
& -\theta \left[ 
\begin{array}{c}
e^{(D,0)}\mathcal{T}_{1}^{\left( \xi _{D}^{2}\right) }\left( u\right) \\ 
+\sum_{k=0}^{\infty }e^{(D,0,k)}\left( \beta _{D},\alpha _{D}\right) 
\mathcal{T}_{1}^{\left( \alpha _{D}+k\right) }\left( u\right) \\ 
+\sum_{k=0}^{\infty }e^{(D,0,k)}\left( \alpha _{D},\beta _{D}\right) 
\mathcal{T}_{1}^{\beta _{D}+k}\left( u\right)%
\end{array}%
\right] ,  \notag
\end{align}%
where 
\begin{align}
\mathcal{T}_{1}^{\left( \tau \right) }\left( u\right) & =\left( \frac{%
	\Upsilon _{D}}{\mu _{1}^{(D)}}\right) ^{\tau }\int_{y=\gamma _{th}}^{\infty }%
\frac{\mathcal{J}_{1}\left( y,u\right) }{y^{-\tau }}  \label{T_tau_expr1} \\
& \times G_{2,4}^{3,1}\left( \frac{\Upsilon _{E_{2}}y}{\mu _{1}^{(E_{2})}}%
\left\vert 
\begin{array}{c}
1-\tau ;\xi _{E_{2}}^{2}+1 \\ 
\xi _{E_{2}}^{2},\alpha _{E_{2}},\beta _{E_{2}};-\tau%
\end{array}%
\right. \right) .  \notag
\end{align}%
Using (\ref{J1_final}), the terms $\mathcal{T}_{1}^{\left( \tau \right)
}\left( u\right) $ can be expressed as%
\begin{align}
\mathcal{T}_{1}^{\left( \tau \right) }\left( u\right) & =\Delta _{SR}\left( 
\frac{\Upsilon _{D}}{\mu _{1}^{(D)}}\right) ^{\tau }\sum_{n_{3}=0}^{m_{SR}-1}%
\frac{\phi _{SR}^{(n_{3})}\Theta _{1}^{\left( n_{3},\tau \right) }(u)}{%
	\Omega ^{n3+1}\left( u\right) }  \notag \\
& -\Delta _{SR}\Delta _{S_{J}E_{1}}\Delta _{SE_{1}}\left( \frac{\Upsilon _{D}%
}{\mu _{1}^{(D)}}\right) ^{\tau }  \notag \\
& \times
\sum_{n_{1}=0}^{m_{S_{J}E_{1}}-1}\sum_{n_{2}=0}^{m_{SE_{1}}-1}%
\sum_{p=0}^{n_{2}}\sum_{n_{3}=0}^{m_{SR}-1}\frac{\mathcal{B}%
	^{(n_{1},n_{2},n_{3},p)}}{\Omega ^{n_{2}-p+n_{3}+1}\left( u\right) }  \notag
\\
& \times \left[ 
\begin{array}{c}
F_{g_{S_{J}P}}\left( \sigma _{S_{J}}\right) \Theta _{2}^{\left(
	n_{1},n_{2},n_{3},p,\tau \right) }(u){} \\ 
+{}\Theta _{3}^{\left( n_{1},n_{2},n_{3},p,\tau \right) }(u)%
\end{array}%
\right] ,  \label{T_tau_exp2}
\end{align}%
where%
\begin{align}
\Theta _{1}^{\left( n_{3},\tau \right) }\left( u\right) & =\int_{\gamma
	_{th}}^{\infty }\frac{e^{-\frac{\upsilon _{SR}}{\Omega \left( u\right) }y}}{%
	y^{-n_{3}-\tau }}  \label{Theta1_exp1} \\
& \times {\small G}_{2,4}^{3,1}\left( \frac{{\small \Upsilon }_{E_{2}}%
	{\small y}}{{\small \mu }_{1}^{(E_{2})}}\left\vert 
\begin{array}{c}
{\small 1-\tau };{\small \xi }_{E_{2}}^{2}{\small +1} \\ 
{\small \xi }_{E_{2}}^{2}{\small ,\alpha }_{E_{2}}{\small ,\beta }_{E_{2}};%
{\small -\tau }%
\end{array}%
\right. \right) {\small dy},  \notag
\end{align}%
\begin{align}
\Theta _{2}^{\left( n_{1},n_{2},n_{3},p,\tau \right) }\left( u\right) &
=\int_{\gamma _{th}}^{\infty }\frac{e^{-\frac{\zeta }{\Omega \left( u\right) 
		}y}}{y^{-n_{2}-n_{3}+p-\tau }}  \notag \\
& \times {\small G}_{2,4}^{3,1}\left( \frac{{\small \Upsilon }_{E_{2}}%
	{\small y}}{{\small \mu }_{1}^{(E_{2})}}\left\vert 
\begin{array}{c}
{\small 1-\tau };{\small \xi }_{E_{2}}^{2}{\small +1} \\ 
{\small \xi }_{E_{2}}^{2}{\small ,\alpha }_{E_{2}}{\small ,\beta }_{E_{2}}%
{\small ;-\tau }%
\end{array}%
\right. \right)  \notag \\
& \times {\small G}_{2,2}^{1,2}\left( \frac{\upsilon _{S_{J}E_{1}}\Omega
	\left( u\right) }{{\small \upsilon }_{SE_{1}}\overline{\gamma }_{S_{J}}%
	{\small y}}\left\vert 
\begin{array}{c}
{\small -p,1;-} \\ 
{\small n}_{1}{\small +1;0}%
\end{array}%
\right. \right) dy,  \label{Theta2_exp1}
\end{align}%
and

\begin{align}
\Theta _{3}^{\left( n_{1},n_{2},n_{3},p,\tau \right) }\left( u\right) &
=\int_{y=\gamma _{th}}^{\infty }\frac{e^{-\frac{\zeta }{\Omega \left(
			u\right) }y}}{y^{-n_{2}-n_{3}+p-\tau }}\mathcal{G}_{2}\left( y,u\right)
\label{Theta3_exp1} \\
& \times {\small G}_{2,4}^{3,1}\left( \frac{{\small \Upsilon }_{E_{2}}%
	{\small y}}{{\small \mu }_{1}^{(E_{2})}}\left\vert 
\begin{array}{c}
{\small 1-\tau };{\small \xi }_{E_{2}}^{2}{\small +1} \\ 
{\small \xi }_{E_{2}}^{2}{\small ,\alpha }_{E_{2}}{\small ,\beta }_{E_{2}}%
{\small ;-\tau }%
\end{array}%
\right. \right) {\small dy},  \notag
\end{align}

Using (\ref{residd}), (\ref{Theta1_exp1}) can be expressed as 
\begin{align}
\Theta _{1}^{\left( n_{3},\tau \right) }\left( u\right) & =e^{(E_{2},\tau
	)}\Lambda _{1}^{\left( n_{3},\xi _{E_{2}}^{2},\tau \right) }\left( u\right)
\label{Theta1_final} \\
& +\sum_{k=0}^{\infty }e^{(E_{2},\tau ,k)}\left( {\small \beta }_{E_{2}}%
{\small ,\alpha }_{E_{2}}\right) \Lambda _{1}^{\left( n_{3},\alpha
	_{E_{2}}+k,\tau \right) }\left( u\right)  \notag \\
& +\sum_{k=0}^{\infty }e^{(E_{2},\tau ,k)}\left( {\small \alpha }_{E_{2}}%
{\small ,\beta }_{E_{2}}\right) \Lambda _{1}^{\left( n_{3},\beta
	_{E_{2}}+k,\tau \right) }\left( u\right) ,  \notag
\end{align}%
where%
\begin{align}
\Lambda _{1}^{\left( n_{3},a,\tau \right) }\left( u\right) & =\left( \frac{%
	\Upsilon _{E_{2}}}{\mu _{1}^{(E_{2})}}\right) ^{a}\left( \frac{\Omega \left(
	u\right) }{\upsilon _{SR}}\right) ^{\tau +n_{3}+a+1} \\
& \times \Gamma \left( \tau +n_{3}+a+1,\frac{\gamma _{th}\upsilon _{SR}}{%
	\Omega \left( u\right) }\right) ,  \notag
\end{align}%
$g(a,u,\tau )=$ and $a=\{\xi _{E_{2}}^{2},\alpha _{E_{2}}+k,\beta
_{E_{2}}+k\}.$

Similarly to (\ref{Theta1_final}), (\ref{Theta2_exp1}) can be expressed as 
\begin{align}
\Theta _{2}^{\left( n_{1},n_{2},n_{3},p,\tau \right) }\left( u\right) &
=e^{(E_{2},\tau )}\mathcal{M}_{1}^{\left( n_{1,}n_{2},n_{3},p,\xi
	_{E_{2}}^{2},\tau \right) }\left( u\right)   \notag \\
& +\sum_{k=0}^{\infty }\Lambda _{2}^{\left( \alpha _{E_{2}}+k\right) }\left(
u,\beta _{E_{2}},\alpha _{E_{2}}\right)   \notag \\
& +\sum_{k=0}^{\infty }\Lambda _{2}^{\left( \beta _{E_{2}}+k\right) }\left(
u,\alpha _{E_{2}},\beta _{E_{2}}\right) ,  \label{Theta2_final}
\end{align}%
where%
\begin{equation}
\Lambda _{2}^{\left( a\right) }\left( u,x,y\right) =e^{(E_{2},\tau
	,k)}\left( x,y\right) \mathcal{M}_{1}^{(n_{1,}n_{2},n_{3},p,a,\tau )}\left(
u\right) ,
\end{equation}%
\begin{equation}
\mathcal{M}_{1}^{({\small n}_{1,}{\small n}_{2}{\small ,n}_{3}{\small %
		,p,a,\tau })}\left( u\right) =\left( \frac{\Upsilon _{E_{2}}}{\mu
	_{1}^{(E_{2})}}\right) ^{a}\left( \frac{\Omega \left( u\right) }{\zeta }%
\right) ^{\varkappa }\mathcal{H}_{1}\left( u\right) ,  \label{M_exp1}
\end{equation}%
with $\varkappa =n_{2}+n_{3}-p+1+a+\tau ,$ $\mathcal{H}_{1}\left( u\right)
=G_{2,3}^{2,2}\left( \frac{{\small \chi }}{\overline{{\small \gamma }}_{%
		{\small S}_{J}}}\left\vert 
\begin{array}{c}
\left( {\small -p,0}\right) ,\left( {\small 1,0}\right) ;- \\ 
\left( {\small n}_{1}{\small +1,0}\right) ,\left( \varkappa ,\frac{\zeta
	\gamma _{th}}{\Omega \left( u\right) }\right) ;\left( {\small 0,0}\right) 
\end{array}
\right. \right) $, and $\chi $ is defined in \textbf{Theorem \ref{theorem1}}.

Analogously to (\ref{Theta2_exp1}) and (\ref{Theta3_exp1}), the term $%
\Theta _{3}^{\left( n_{2},n_{3},p\right) }\left( u\right) $ can be evaluated
as 
\begin{align}
\Theta _{3}^{\left( n_{1},n_{2},n_{3},p,\tau \right) }\left( u\right) &
=e^{(E_{2},\tau )}\Phi _{1}^{\left( n_{1,}n_{2},n_{3},p,\xi
	_{E_{2}}^{2},\tau \right) }\left( u\right)   \label{Theta3_final} \\
& +\sum_{k=0}^{\infty }\Lambda _{3}^{\left( \alpha _{E_{2}}+k\right) }\left(
u,\beta _{E_{2}},\alpha _{E_{2}}\right)   \notag \\
& +\sum_{k=0}^{\infty }\Lambda _{3}^{\left( \beta _{E_{2}}+k\right) }\left(
u,\alpha _{E_{2}},\beta _{E_{2}}\right) ,  \notag
\end{align}%
where%
\begin{equation}
\Lambda _{3}^{\left( a\right) }\left( u,x,y\right) =e^{(E_{2},\tau
	,k)}\left( x,y\right) {\small \Phi }_{1}^{\left( {\small n}_{1,}{\small n}%
	_{2}{\small ,n}_{3}{\small ,p,a,\tau }\right) }\left( u\right) ,
\end{equation}%
\begin{align}
{\small \Phi }_{1}^{\left( {\small n}_{1,}{\small n}_{2}{\small ,n}_{3}%
	{\small ,p,a,\tau }\right) }\left( u\right) & =\left( \frac{\Upsilon _{E_{2}}%
}{\mu _{1}^{(E_{2})}}\right) ^{a}\left( \frac{\Omega \left( u\right) }{\zeta 
}\right) ^{\varkappa }  \label{Phia_exp1} \\
& \times {\small G}_{3,3}^{2,3}\left( {\small b}\left\vert 
\begin{array}{c}
\left( {\small -p,0}\right) ,\left( {\small 1,0}\right) ,\left( {\small 0,c}%
\right) ;- \\ 
\left( {\small n}_{1}{\small +1,0}\right) ,\left( {\small \varkappa ,v}%
\right) ;\left( {\small 0,0}\right) 
\end{array}%
\right. \right) .  \notag
\end{align}%
where $b=\frac{\chi }{\lambda _{S_{J}P}\overline{\gamma }_{I}}$, $c={\small %
	\sigma }_{S_{J}}{\small \lambda }_{S_{J}P},$ and $v=\frac{\zeta \gamma _{th}%
}{\Omega \left( u\right) }.$

{Now, the remaining last step in this proof consists of computing the
	expression for the IP in the presence of a friendly jammer by incorporating $%
	\left( \ref{I1_exp2}\right) $ into $\left( \ref{IP_new_expression}\right) $
	as} 
\begin{align}
\text{IP}_{P}& =1-\frac{\xi _{E_{2}}^{2}\mathcal{D}_{P}^{\left( 0\right) }}{%
	\Gamma \left( \alpha _{E_{2}}\right) \Gamma \left( \beta _{E_{2}}\right) }
\label{IP_exp3} \\
& +\theta \left[ 
\begin{array}{c}
e^{(D,0)}\mathcal{D}_{P}^{\left( \xi _{D}^{2}\right) } \\ 
+\sum_{k=0}^{\infty }e^{(D,0,k)}\left( \beta _{D},\alpha _{D}\right) 
\mathcal{D}_{P}^{\left( \alpha _{D}+k\right) } \\ 
+\sum_{k=0}^{\infty }e^{(D,0,k)}\left( \alpha _{D},\beta _{D}\right) 
\mathcal{D}_{P}^{\left( \beta _{D}+k\right) }%
\end{array}%
\right] ,  \notag
\end{align}%
where 
\begin{equation}
\mathcal{D}_{P}^{\left( \tau \right) }=\frac{\Upsilon _{D}^{\tau }}{\left(
	\mu _{1}^{(D)}\right) ^{\tau }}\int_{u=0}^{\infty }f_{g_{SP}}\left( u\right) 
\mathcal{T}_{1}^{\left( \tau \right) }\left( u\right) du.
\label{Dtau_jammer}
\end{equation}%
To compute $\left( \ref{Dtau_jammer}\right) ,$ we need to replace $\Omega
\left( u\right) $ by its values in (\ref{T_tau_exp2}). Therefore, when $%
u\leq \sigma _{S},$ the term $\mathcal{T}_{1}^{\left( \tau \right) }\left(
u\right) $ becomes constant. It follows that%
\begin{equation}
\mathcal{D}_{P}^{\left( \tau \right) }=\frac{\mathcal{T}_{2}^{\left( \tau
		\right) }\left( \overline{\gamma }_{S}\right) F_{g_{SP}}\left( \sigma
	_{S}\right) +\lambda _{SP}\mathcal{T}_{2}^{\left( \tau \right) }\left( 
	\overline{\gamma }_{I}\right) }{\left( \mu _{1}^{(D)}\right) ^{\tau }},
\label{Dtau_jammer_exp2}
\end{equation}%
with $\mathcal{T}_{2}^{\left( \tau \right) }\left( \overline{\gamma }%
_{I}\right) =\int_{\sigma _{S}}^{\infty }e^{-\lambda _{SP}u}\mathcal{T}%
_{1}^{\left( \tau \right) }\left( u\right) ,$ while the term $\mathcal{T}%
_{2}^{\left( \tau \right) }\left( \overline{\gamma }_{S}\right) $ can be
obtained by replacing $\Omega \left( u\right) =\overline{\gamma }_{S}$ in (%
\ref{T_tau_exp2}).

Also, to compute $\mathcal{T}_{2}^{\left( \tau \right) }\left( \overline{%
	\gamma }_{I}\right) $, we need to replace $\Omega \left( u\right) =\frac{%
	\overline{\gamma }_{I}}{u}$ in (\ref{T_tau_exp2}) and compute the following
integrals%
\begin{equation}
\mathcal{S}^{\left( n_{3},a,\tau \right) }=\int_{\sigma _{S}}^{\infty }\frac{%
	e^{-\lambda _{SP}u}}{u^{a+\tau }}\Gamma \left( \tau +n_{3}+a+1,\epsilon
_{I}\upsilon _{SR}u\right) du,  \label{S1_exp1}
\end{equation}%
\begin{equation}
\mathcal{Y}^{(n_{1,}n_{2},n_{3},p,a,\tau )}=\int_{\sigma _{S}}^{\infty }%
\frac{e^{-\lambda _{SP}u}}{u^{a+\tau }}\mathcal{M}%
_{1}^{(n_{1,}n_{2},n_{3},p,a,\tau )}\left( u\right) du,  \label{Y2_exp1}
\end{equation}%
and 
\begin{equation}
\mathcal{W}^{(n_{1,}n_{2},n_{3},p,a,\tau )}=\int_{\sigma _{S}}^{\infty }%
\frac{e^{-\lambda _{SP}u}}{u^{a+\tau }}\Phi
_{1}^{(n_{1,}n_{2},n_{3},p,a,\tau )}\left( u\right) du,  \label{Wa_exp1}
\end{equation}%
where $\epsilon _{I}$ is defined is \textbf{Theorem \ref{theorem1}}.

Using \cite[Eq. 06.06.26.0005.01]{Wolfram}, (\ref{S1_exp1}) can be expressed as
\begin{eqnarray}
\mathcal{S}^{\left( n_{3},a,\tau \right) } &=&\frac{G_{2,2}^{2,1}\left(
	\varepsilon \left\vert 
	\begin{array}{c}
	\left( a+\tau ,\lambda _{SP}\sigma _{S}\right) ;\left( 1,0\right)  \\ 
	\left( 0,0\right) ,\left( \tau +n_{3}+a+1,0\right) ;-%
	\end{array}%
	\right. \right) }{\lambda _{SP}^{-a-\tau +1}},  \notag \\
&&  \label{S1_final}
\end{eqnarray}
with $\varepsilon =\frac{\epsilon _{I}\upsilon _{SR}}{\lambda _{SP}},$ while (\ref{Y2_exp1}) and (\ref{Wa_exp1}) can be evaluated by replacing {$\Omega \left( u\right) =\frac{\overline{\gamma }_{I}}{u}$ into (\ref{M_exp1}) and (\ref{Phia_exp1})}, respectively, and calculating the following common integral

\begin{align}
\mathcal{A}^{(n_{1,}n_{2},n_{3},p,\tau )}& =\int_{\sigma _{S}}^{\infty
}u^{-a-\tau }e^{-\lambda _{SP}u}\Gamma \left( \varkappa +s,\epsilon
_{I}\zeta u\right) du  \label{A_final} \\
& \overset{\left( a\right) }{=}\frac{1}{2\pi j}\int_{\mathcal{C}_{w}}\frac{%
	\Gamma \left( w\right) \Gamma \left( \varkappa +s+w\right) }{\Gamma \left(
	1+w\right) }  \notag \\
& \times \left( \frac{\epsilon _{I}\zeta }{\lambda _{SP}}\right) ^{-w}\frac{%
	\Gamma \left( 1-a-\tau -w,\sigma _{S}\lambda _{SP}\right) }{\lambda
	_{SP}^{-a-\tau +1}}dw,  \notag
\end{align}%
where step $(a)$ follows using \cite[Eq. 06.06.26.0005.01]{Wolfram}.

Finally, substituting (\ref{A_final}) into (\ref{Y2_exp1}) and (\ref{Wa_exp1}%
), one obtains (\ref{Ya_final}) and (\ref{Wa_final}) which concludes the
proof of \textbf{Theorem \ref{theorem1}}.\qedhere

\section*{Appendix C: proof of Theorem 2}

{One can clearly see from (\ref{IP_final}) and (\ref{Dtau_jammer_final})
	that the asymptotic expression for the IP depends on approximating (\ref{Em}%
	), which can be obtained by determining the asymptotic expansion of (\ref%
	{Psi2}) and (\ref{Psi3}). To do so,} the residues theorem is applied to find
the asymptotic expressions of functions given in (\ref{M2})-(\ref{Wa_final}).

First, the functions given in (\ref{M2}) and (\ref{Phi2}) can be rewritten
as Mellin-Barnes integrals as 
\begin{align}
\mathcal{M}^{(n_{1,}n_{2},n_{3},p,a,\tau )}& =\frac{1}{2\pi j}\int_{\mathcal{%
		C}}\frac{\Gamma \left( n_{1}+1+s\right) {\small \Gamma }\left( {\small %
		\varkappa +s},{\small \zeta }\epsilon _{I}\sigma _{S}\right) }{\Gamma \left(
	1-s\right) }  \label{M2_integral} \\
& \times \Gamma \left( 1+p-s\right) \Gamma \left( -s\right) \left( \frac{%
	{\small \chi \sigma }_{S_{J}}}{\overline{\gamma }_{I}}\right) ^{-s}{\small ds%
},  \notag
\end{align}%
and

\begin{align}
\Phi ^{(n_{1,}n_{2},n_{3},p,a,\tau )}& =\frac{1}{2\pi j}\int_{\mathcal{C}}%
\frac{\Gamma \left( n_{1}+1+s\right) \Gamma \left( \varkappa +s,\zeta
	\epsilon _{I}\sigma _{S}\right) }{\Gamma \left( 1-s\right) }
\label{Phi2_integral} \\
& \times \Gamma \left( 1+p-s\right) \Gamma \left( -s\right)   \notag \\
& \times \Gamma \left( 1-s,\sigma _{S_{J}}\lambda _{S_{J}P}\right) \left( 
\frac{\chi }{\lambda _{S_{J}P}\overline{\gamma }_{I}}\right) ^{-s}ds,  \notag
\end{align}%
where $\varkappa =n_{2}+n_{3}-p+1+a+\tau ,$ and {$j=\sqrt{-1},$ $\mathcal{C}%
	_{s}$ is a vertical line of integration chosen such as to separate the left
	poles of the above integrand functions from the right ones,}

It is noteworthy that the same complex contour, namely {$\mathcal{C}_{s}$}
can be used to evaluate both integrals as the upper incomplete Gamma
function has no poles and both integrands have the same poles. Moreover,\
the conditions of \cite[Theorem 1.5]{kilbas} hold. That is, the two above
complex integrals can be written as an infinite sum of the poles belonging
to the left half plan of $\mathcal{C}$. Furthermore, it is clearly seen that
(\ref{M2_integral}) and (\ref{Phi2_integral}) both have same left poles $%
-n_{1}-1-k,$ $k\in 
%TCIMACRO{\U{2115} }%
%BeginExpansion
\mathbb{N}
%EndExpansion
.$ It follows that%
\begin{align}
\mathcal{M}^{(n_{1,}n_{2},n_{3},p,a,\tau )}& =\sum_{k=0}^{\infty }\frac{%
	(-1)^{k}\Gamma \left( 2+p+n_{1}+k\right) }{k!\left( n_{1}+1+k\right) 
	\overline{\gamma }_{I}^{n_{1}+1+k}}  \notag \\
& \times \frac{\Gamma \left( \varkappa -n_{1}-1-k,\zeta \epsilon _{I}\sigma
	_{S}\right) }{\left( \chi \sigma _{S_{J}}\right) ^{-n_{1}-1-k}},
\label{M2_residue}
\end{align}%
and 
\begin{align}
\Phi ^{(n_{1,}n_{2},n_{3},p,a,\tau )}& =\sum_{k=0}^{\infty }\frac{%
	(-1)^{k}\Gamma \left( 2+p+n_{1}+k\right) }{k!\left( n_{1}+1+k\right) \left(
	\lambda _{S_{J}P}\overline{\gamma }_{I}\right) ^{n_{1}+1+k}}
\label{Phi2_residue} \\
& \times \Gamma \left( \varkappa -n_{1}-1-k,\zeta \epsilon _{I}\sigma
_{S}\right)  \notag \\
& \times \Gamma \left( 2+n_{1}+k,\sigma _{S_{J}}\lambda _{S_{J}P}\right)
\chi ^{n_{1}+1+k}.  \notag
\end{align}%
By considering only the first term of the infinite summation when $\overline{%
	\gamma }_{I}\rightarrow \infty $, (\ref{M2_residue}) and (\ref{Phi2_residue}%
) can be asymptotically approximated by

\begin{align}
\mathcal{M}^{(n_{1,}n_{2},n_{3},p,a,\tau )}& \sim \frac{\Gamma \left(
	\varkappa -n_{1}-1-k,\zeta \epsilon _{I}\sigma _{S}\right) }{n_{1}+1}  \notag
\\
& \times \Gamma \left( 2+p+n_{1}\right) \left( \frac{\chi \sigma _{S_{J}}}{%
	\overline{\gamma }_{I}}\right) ^{n_{1}+1},  \label{M2_approximate}
\end{align}

\begin{align}
\Phi ^{(n_{1,}n_{2},n_{3},p,a,\tau )}& \sim \frac{\Gamma \left(
	2+n_{1},\sigma _{S_{J}}\lambda _{S_{J}P}\right) }{n_{1}+1}
\label{Phi2_approximate} \\
& \times \Gamma \left( 2+p+n_{1}\right) \left( \frac{\chi }{\lambda _{S_{J}P}%
	\overline{\gamma }_{I}}\right) ^{n_{1}+1}  \notag \\
& \times \Gamma \left( \varkappa -n_{1}-1,\zeta \epsilon _{I}\sigma
_{S}\right) .  \notag
\end{align}%
In similar manner to (\ref{M2_residue}) and (\ref{Phi2_residue}), (\ref%
{Ya_final}) and (\ref{Wa_final}) can be, respectively, expressed as infinite
sums as follows 
\begin{align}
\mathcal{Y}^{(n_{1,}n_{2},n_{3},p,a,\tau )}& =\frac{\lambda _{SP}^{a+\tau -1}%
}{2\pi j}\int_{\mathcal{C}_{w}}\frac{\Gamma \left( {\small 1-a-\tau -w},%
	{\small \sigma }_{S}{\small \lambda }_{SP}\right) }{w\left( \frac{{\small %
			\zeta \epsilon }_{I}}{{\small \lambda }_{SP}}\right) ^{w}}  \notag \\
& \times \left[ 
\begin{array}{c}
\sum_{k=0}^{\infty }\frac{(-1)^{k}\Gamma \left( 2+p+n_{1}+k\right) }{%
	k!\left( n_{1}+1+k\right) \overline{\gamma }_{I}^{n_{1}+1+k}} \\ 
\times \frac{{\small \Gamma }\left( {\small \varkappa +w-n}_{1}{\small -1-k}%
	\right) }{\left( \chi \sigma _{S_{J}}\right) ^{-n_{1}-1-k}} \\ 
+\frac{\left( -1\right) ^{k}\Gamma \left( n_{1}+1-\varkappa -w-k\right) }{%
	k!\left( \varkappa +w+k\right) \overline{\gamma }_{I}^{\varkappa +w+k}} \\ 
\times \frac{\Gamma \left( 1+p+\varkappa +w+k\right) }{\left( \chi \sigma
	_{S_{J}}\right) ^{-\varkappa -w-k}}%
\end{array}%
\right] dw,  \notag \\
&  \label{Ya_residue}
\end{align}%
and%
\begin{align}
\mathcal{W}^{(n_{1,}n_{2},n_{3},p,a,\tau )}& =\int_{\mathcal{C}_{w}}\frac{%
	\Gamma \left( 1-a-\tau -w,\sigma _{S}\lambda _{SP}\right) }{w}\left( \frac{%
	\zeta \epsilon _{I}}{\lambda _{SP}}\right) ^{-w}  \notag \\
& \times \left[ 
\begin{array}{c}
\sum_{k=0}^{\infty }\frac{(-1)^{k}\Gamma \left( 2+p+n_{1}+k\right) }{%
	k!\left( n_{1}+1+k\right) } \\ 
\Gamma \left( \varkappa +w-n_{1}-1-k\right) \\ 
\times \Gamma \left( {\small 2+n}_{1}{\small +k},{\small \sigma }_{S_{J}}%
{\small \lambda }_{S_{J}P}\right) \\ 
\times \left( \frac{\chi }{\lambda _{S_{J}P}\overline{\gamma }_{I}}\right)
^{n_{1}+1+k} \\ 
+\frac{\left( -1\right) ^{k}\Gamma \left( n_{1}+1-\varkappa -w-k)\right) }{%
	k!\left( \varkappa +w+k\right) } \\ 
\times \Gamma \left( 1+p+\varkappa +w+k\right) \\ 
\times \Gamma \left( 1+\varkappa +w+k,\sigma _{S_{J}}\lambda _{S_{J}P}\right)
\\ 
\times \left( \frac{\chi }{\lambda _{S_{J}P}\overline{\gamma }_{I}}\right)
^{\varkappa +w+k}%
\end{array}%
\right] dw,  \notag \\
&  \label{Wa_residue}
\end{align}

Subsequently, their asymptotic expression in high SNR regime can be
straightforwardly obtained by taking the first term of the two above
infinite summations as

\begin{align}
\mathcal{Y}^{(n_{1,}n_{2},n_{3},p,a,\tau )}& \sim \frac{\lambda
	_{SP}^{a+\tau -1}\Gamma \left( 2+p+n_{1}\right) }{n_{1}+1}\left( \frac{\chi
	\sigma _{S_{J}}}{\overline{\gamma }_{I}}\right) ^{n_{1}+1}
\label{Ya_approximate} \\
& \times {\small G}_{2,2}^{2,1}\left( \frac{{\small \zeta \epsilon }_{I}}{%
	{\small \lambda }_{SP}}\left\vert 
\begin{array}{c}
\left( {\small a+\tau ,\sigma }_{S}{\small \lambda }_{SP}\right) ;\left( 
{\small 1,0}\right) \\ 
\left( 0,0\right) ,\left( {\small \varkappa -n}_{1}{\small -1,0}\right) ;-%
\end{array}%
\right. \right) ,  \notag
\end{align}%
and

\begin{align}
\mathcal{W}^{(n_{1,}n_{2},n_{3},p,a,\tau )}& \sim \frac{\lambda
	_{SP}^{a+\tau -1}\Gamma \left( 2+p+n_{1}\right) }{n_{1}+1}
\label{Wa_approximate} \\
& \times \Gamma \left( 2+n_{1},\sigma _{S_{J}}\lambda _{S_{J}P}\right)
\left( \frac{\chi }{\lambda _{S_{J}P}\overline{\gamma }_{I}}\right)
^{n_{1}+1}  \notag \\
& \times {\small G}_{2,2}^{2,1}\left( \frac{{\small \zeta \epsilon }_{I}}{%
	{\small \lambda }_{SP}}\left\vert 
\begin{array}{c}
\left( {\small a+\tau ,\sigma }_{S}{\small \lambda }_{SP}\right) ;\left( 
{\small 1,0}\right) \\ 
\left( 0,0\right) ,\left( {\small \varkappa -n}_{1}{\small -1,0}\right) ;-%
\end{array}%
\right. \right) .  \notag
\end{align}%
Lastly, substituting (\ref{M2_approximate}) alongside (\ref{Ya_approximate}%
) and (\ref{Phi2_approximate}) along with (\ref{Wa_approximate}) into (\ref%
{Psi2}) and (\ref{Psi3}), respectively, (\ref{Psi2_approximate}) and (\ref%
{Psi3_approximate}) are attained. This concludes the proof of \textbf{Theorem \ref{theorem2}}.\qedhere

\ifCLASSOPTIONcaptionsoff
  \newpage
\fi

%\begin{IEEEbiography}{Yuguang ``Michael'' Fang}
%Biography text here.
%\end{IEEEbiography}

%It is not necessary to upload the biography when you submit your manuscript.

\end{document}